\documentclass[journal,onecolumn,12pt,draftclsnofoot]{IEEEtran}
\usepackage{amsmath}
\usepackage{mathrsfs}

\usepackage{amsfonts}

\usepackage{graphicx}
\usepackage{citesort}
\usepackage{amssymb}
\usepackage{amsfonts}
\usepackage{amsmath}
\usepackage{epsfig}
\usepackage{color}
\usepackage{fancybox}
\usepackage{textcomp}
\usepackage{multirow}
\usepackage{setspace}
\usepackage{psfrag}

\newtheorem{theorem}{Theorem}
\newtheorem{lemma}{Lemma}

\newtheorem{assumption}{Assumption}

\newtheorem{corollary}{Corollary}
\newtheorem{proposition}{Proposition}

    \def\Complex{{\rm\rule[.23ex]{.03em}{1.1ex}\kern-.3em{C}}}

    \newcommand{\be}{\begin{equation}} \newcommand{\ee}{\end{equation}}
    \newcommand{\bea}{\begin{eqnarray}} \newcommand{\eea}{\end{eqnarray}}
    \newcommand{\benum}{\begin{enumerate}} \newcommand{\eenum}{\end{enumerate}}

    \newcommand{\qf}{{\bf f}}
    \newcommand{\qg}{{\bf g}}
    \newcommand{\qh}{{\bf h}}

    \newcommand{\qq}{{\bf q}}
    
    \newcommand{\qs}{{\bf s}}

    \newcommand{\qw}{{\bf w}}
    
    \newcommand{\qy}{{\bf y}}
    \newcommand{\qz}{{\bf z}}

    \newcommand{\qA}{{\bf A}}
    \newcommand{\qB}{{\bf B}}
    
    \newcommand{\qD}{{\bf D}}
    
    \newcommand{\qF}{{\bf F}}
    \newcommand{\qG}{{\bf G}}
    \newcommand{\qH}{{\bf H}}
    \newcommand{\qI}{{\bf I}}

    \newcommand{\qP}{{\bf P}}
    \newcommand{\qQ}{{\bf Q}}
    \newcommand{\qR}{{\bf R}}
    
    \newcommand{\qT}{{\bf T}}

    \newcommand{\qW}{{\bf W}}
    \newcommand{\qX}{{\bf X}}
    \newcommand{\qY}{{\bf Y}}
    
    \newcommand{\qzero}{{\bf 0}}

    \newcommand{\te}{{\tilde{e}}}

    \newcommand{\tqH}{{\tilde{\qH}}}
    \newcommand{\tqF}{{\tilde{\qF}}}
    \newcommand{\tqW}{{\tilde{\qW}}}

    \newcommand{\tqh}{{\tilde{\qh}}}
    \newcommand{\tqf}{{\tilde{\qf}}}
    
    \newcommand{\tqw}{{\tilde{\qw}}}

    \newcommand{\bbR}{{\mathbb R}}
    \newcommand{\bbC}{{\mathbb C}}

    \newcommand{\calN}{{\mathcal N}}

    \newcommand{\diag}{{\sf diag}}
    
    \newcommand{\tr}{{\sf tr}}
    \newcommand{\Ex}{{\sf E}}

    \newcommand{\argmax}{\operatornamewithlimits{arg\, max}}

\begin{document}

\title{Large System Analysis of Cognitive Radio Network via Partially-Projected Regularized Zero-Forcing Precoding}

%
%
%

\author{Jun Zhang\thanks{J. Zhang is with Jiangsu Key Laboratory of Wireless Communications, Nanjing University of Posts and Telecommunications, Nanjing 210003, P. R. China, and also with Singapore University of Technology and Design, Singapore 487372, Email address: {\sf zhang\_jun@sutd.edu.sg}.},  Chao-Kai Wen\thanks{C. K. Wen is with the Institute of Communications Engineering, National Sun Yat-sen University, Kaohsiung 804, Taiwan, E-mail address: {\sf ckwen@ieee.org}.}, Chau Yuen\thanks{C. Yuen is with Singapore University of Technology and Design, Singapore 487372, Email address: {\sf yuenchau@sutd.edu.sg}.}, Shi Jin\thanks{S. Jin and X. Q. Gao are with the National Mobile Communications Research Laboratory, Southeast University, Nanjing 210096, China, E-mail addresses: {\sf \{jinshi, xqgao\}@seu.edu.cn}.}, and Xiqi Gao}

\maketitle

\begin{abstract}


In this paper, we consider a cognitive radio (CR) network in which a secondary multiantenna base station (BS) attempts to communicate with multiple secondary
users (SUs) using the radio frequency spectrum that is originally allocated to multiple primary users (PUs). Here, we employ partially-projected regularized
zero-forcing (PP-RZF) precoding to control the amount of interference at the PUs and to minimize inter-SUs interference. The PP-RZF precoding
partially projects the channels of the SUs into the null space of the channels from the secondary BS to the PUs. The regularization parameter and the projection
control parameter are used to balance the transmissions to the PUs and the SUs. However, the search for the optimal parameters, which can maximize the ergodic
sum-rate of the CR network, is a demanding process because it involves Monte-Carlo averaging. Then, we derive a deterministic expression for the ergodic sum-rate
achieved by the PP-RZF precoding using recent advancements in large dimensional random matrix theory. The deterministic equivalent enables us to efficiently
determine the two critical parameters in the PP-RZF precoding because no Monte-Carlo averaging is required. Several insights are also obtained through the analysis.

\end{abstract}


\section{Introduction}

The radio frequency spectrum is a valuable but congested natural resource because it is shared by an increasing number of users. Cognitive radio (CR) \cite{Mitola-99PCM,Letaief-09ProcIEEE,Goldsmith-09IProc,Liang-11TVT} is viewed as an effective means to improve the utilization of the
radio frequency spectrum by introducing dynamic spectrum access technology. Such technology allows secondary users (SUs, also known as CR users) to access the
radio spectrum originally allocated to primary users (PUs). In the CR literature, two cognitive spectrum access models have been widely adopted
\cite{Liang-11TVT}: 1) the \emph{opportunistic spectrum access}  model and 2) the \emph{concurrent spectrum access} model. In the opportunistic spectrum access
model, SUs carry out spectrum sensing to detect spectrum holes and reconfigure their transmission to operate only in the identified holes
\cite{Mitola-99PCM,Kim-11TVT}. Meanwhile, in the concurrent spectrum access model, SUs transmit simultaneously with PUs as long as interference to PUs is limited
\cite{Zheng-09TSP,Huang-11JSAC}.

In this paper, we focus on the concurrent spectrum access model particularly when the secondary base station (BS) is equipped with multiple antennas. A desirable condition in the concurrent spectrum access model is for SUs to maximize their own performance while minimizing the interference caused to the PUs.
Several transmit schemes have been studied to balance the transmissions to the SUs and the PUs \cite{ZhangR-08JSTSP,Hamdi-09TWC,HeYY-13TWC,HeYY-13ICASSP,Chen-13TVT}. In \cite{ZhangR-08JSTSP}, a transmit algorithm has been proposed based on the singular value decomposition of the secondary channel after the projection into the null
space of the channel from the secondary BS to the PUs. A spectrum sharing scheme has been designed for a large number of SUs \cite{Hamdi-09TWC}, in which the
SUs are pre-selected so that their channels are nearly orthogonal to the channels of the PUs. Doing so ensures that the SUs cause the lowest interference to the PUs.

In multi-antenna and multiuser downlink systems, a common technique to mitigate the multiuser interference is a zero-forcing (ZF) precoding
\cite{Caire-03TIT,Samardzija-07ICC,Irmer-09COMMag,Suraweera-13ICC}, which is computationally more efficient than its non-linear alternatives. However, the achievable rates of the ZF precoding are severely compromised when the channel matrix is ill conditioned. Then, regularized ZF (RZF) precoding \cite{Joham-02ISSSTA,Peel-05Tcom} is proposed to mitigate the ill-conditioned problem by employing a regularization parameter in the channel inversion. The regularization parameter can control the amount of introduced interference. Several applications based on the RZF framework have been developed, such as transmitter designs for \emph{non}-CR broadcast systems \cite{Nguyen-08GLCOM,Muharar-11ICC,Wagner-12IT,Muharar-13TCom}, security systems \cite{Geraci-13JSAC,ZhangJun-14CL}, and multi-cell cooperative systems
\cite{Muharar-12ISIT,Huang-13TWC,ZhangJun-13TWC,WenCK-14TWC}.

While directly applying RZF to CR networks, the secondary BS can \emph{only} control the interference in inter-SUs. A \emph{partially-projected} RZF (PP-RZF)
precoding has been proposed \cite{HeYY-13TWC,HeYY-13ICASSP}, which limits the interference from the SUs to the PUs by combining the RZF
\cite{Joham-02ISSSTA,Peel-05Tcom} with the channel projection idea \cite{ZhangR-08JSTSP}. The PP-RZF precoding follows the classical RZF technique, although the
former is based on the partially-projected channel, which is obtained by partially projecting the channel matrix into the null space of the channel from the
secondary BS to the PUs. The amount of interference to the PUs decreases with increasing amounts of projection into the null space of the PUs, which can be
achieved by tuning the projection control parameter. However, the search for the optimal regularization parameter and projection control parameter is a demanding
process because it involves Monte-Carlo averaging. Therefore, a deterministic (or large-system) approximation of the signal-to-interference-plus-noise ratio (SINR)
for the PP-RZF scheme has been derived \cite{HeYY-13TWC,HeYY-13ICASSP}. Unfortunately, only the CR channel with a \emph{single} PU has been studied and the scenario where \emph{multiple} PUs are present remains unsolved \cite{HeYY-13TWC}.

To apply the PP-RZF precoding scheme in a CR network with \emph{multiple} PUs, a new analytical technique that deals with a \emph{multi}-dimensional random projection matrix, which is generated by partially projecting the channel matrix into the null spaces of \emph{multiple} PUs, is required. This paper aims to address the above mentioned challenge by providing analytical results in a more general setting than that in \cite{HeYY-13TWC,HeYY-13ICASSP}. Specifically, we focus on a downlink multiuser CR network (Fig.~\ref{fig:1}), which consists of a secondary BS with multiple antennas, SUs, and PUs as well as different channel gains. Our main contributions are summarized below.

\begin{itemize}
\item We derive deterministic equivalents for the SINR and the ergodic sum-rate achieved by the PP-RZF precoding under the general CR network. Unlike previous works \cite{HeYY-13TWC,HeYY-13ICASSP}, our model considers \emph{multiple} PUs and allows different channel gains from the secondary BS to each user. Owing to recent advancements in large dimensional random matrix theory (RMT) with respect to complex combinations of different types of
    independent random matrices \cite{Couillet-11BOOK}, we identify the large system distribution of the Stieltjes transform for a new class of random matrix. Therefore, our extension becomes non trivial and novel.

\item In the PP-RZF precoding, the regularization parameter and the projection control parameter can regulate the amount of interference to the SUs and the PUs, but a wrong choice of parameters can considerably degrade the performance of the CR network. However, the search for the optimal parameters is a demanding process because Monte-Carlo averaging is required. We overcome the fundamental difficulty of applying PP-RZF precoding in the CR network. The deterministic equivalent for the ergodic sum-rate provides an efficient way of finding the asymptotically optimal regularization parameter and the asymptotically optimal projection control parameter. Simulation results indicate good agreement with the optimum in terms of the ergodic sum-rate.

\item We provide several useful observations on the condition that the regularization parameter and the projection control parameter can achieve the optimal sum-rate. We also reveal the relationship between the parameters and the signal-to-noise ratio (SNR).
\end{itemize}

\emph{Notations}---We use uppercase and lowercase boldface letters to denote matrices and vectors, respectively.
An $N \times N$ identity matrix is denoted by ${\bf{I}}_N$, an all-zero matrix by ${\bf 0}$, and an all-one matrix by ${\bf 1}$.
The superscripts $(\cdot)^{H}$, $(\cdot)^{T}$, and $(\cdot)^{*}$ denote the conjugate transpose, transpose, and conjugate operations, respectively. $\Ex\{\cdot\}$ returns the expectation with respect to all random variables within the bracket, and $\log(\cdot)$ is the natural logarithm.
We use $[{\bf A}]_{kl}$, $[{\bf A}]_{l,k}$, or $A_{kl}$ to denote the ($k$,$l$)-th entry of the matrix $\bf A$, and $a_k$ denotes the $k$-th entry of the column vector $\bf{a}$.
The operators $(\cdot)^{\frac{1}{2}}$, $(\cdot)^{-1}$, ${\tr}(\cdot)$, and $\det(\cdot)$ represent the matrix principal square root, inverse, trace, and determinant, respectively, $\|\cdot\|$ represents the Euclidean norm of an input vector or the spectral norm of an input matrix, and $\diag(\bf{x})$ denotes a diagonal matrix with $\bf{x}$ along its main diagonal. The notation ``$\xrightarrow{a.s.}$'' denotes the almost sure (a.s.) convergence.

\section{System Model and Problem Formulation}
\subsection{System Model}

\begin{figure}
\begin{center}
\resizebox{4.5in}{!}{%
\includegraphics*{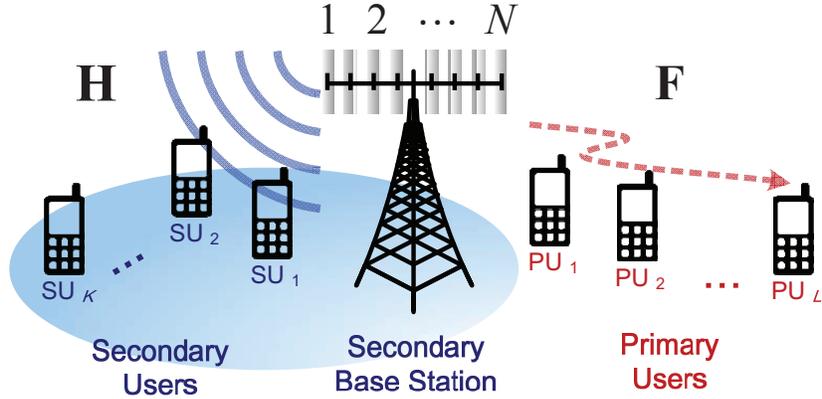} }%
\caption{A downlink multiuser cognitive radio network.}\label{fig:1}
\end{center}
\end{figure}

As illustrated in Fig.~\ref{fig:1}, we consider a downlink multiuser CR network that consists of a secondary BS with $N$ antennas (labeled as ${\sf BS}$). The ${\sf BS}$ simultaneously transmits $K$ independent messages to $K$ single antenna SUs (labeled as ${\sf SU}_1, \dots, {\sf SU}_K$). We assume that
all the SUs share the same spectrum with $L$ single antenna PUs (labeled as ${\sf PU}_1, \dots, {\sf PU}_L$). Let $\qh_k^H \in \bbC^{1 \times N}$ be the
fading channel vector between ${\sf BS}$ and ${\sf SU}_k$,  $\qf_{l}^H \in \bbC^{1 \times N}$ be the fading channel vector between ${\sf BS}$ and ${\sf PU}_l$, and
$\qg_{k} \in \bbC^{N \times 1}$ be the precoding vector of ${\sf SU}_k$. The received signal at ${\sf SU}_k$ can therefore be expressed as
\begin{equation}\label{eq:the received signal of SUEk}
 y_k = \qh_k^H \qg_k s_k + \sum_{j=1, j \ne k}^K \qh_k^H \qg_j s_j + z_k,
\end{equation}
where $s_k$ is the data symbol of ${\sf SU}_k$, $s_j$'s are independent and identically distributed (i.i.d.) data symbols with zero mean and unit variance, respectively,
and $z_k$ is the additive Gaussian noise with zero mean and variance of $\sigma^2$. For ease of exposition, we define $\qH \triangleq \left[ \qh_1,
\dots, \qh_K \right]^H \in \bbC^{K\times N}$, $\qF \triangleq \left[ \qf_1, \dots, \qf_L \right]^H \in \bbC^{L\times N}$, $\qG \triangleq \left[ \qg_1, \dots,
\qg_K \right] \in \bbC^{N \times K}$, $\qy\triangleq[ {y_1}, \dots ,{y_K}]^T \in \bbC^K$, $\qs \triangleq \left[{s_1}, \dots ,{s_K} \right]^T \in \bbC^K$, and $\qz
\triangleq \left[ {z_1}, \dots ,{z_K} \right]^T \in \bbC^K$. The received signal of all the SUs in vector form is given by
\begin{equation}\label{eq:the concatenated received signal vector}
 \qy= \qH \qG \qs + \qz.
\end{equation}
We also assume that ${\sf BS}$ satisfies the average total transmit power constraint
\begin{equation}\label{eq:base station transmitted power constrain}
 \Ex \left\{ \tr\left( \qG \qG^H \right) \right\} \leq {N P_T},
\end{equation}
where $P_T > 0$ is the parameter that determines the power budget of ${\sf BS}$. Notably, if we consider the instantaneous transmit power constraint, i.e., $\tr( \qG \qG^H ) \leq {N P_T}$, we can obtain the same constraint in a large-system regime, as shown in Appendix B-III.

The peak received interference power constraint or the average received interference power constraint is used to protect the PUs. Given that
the latter is more flexible for dynamically allocating transmission powers over different fading states than the former
\cite{ZhangR-09TWC,Wang-09TWC}, we employ the average received interference power constraint and consider two cases: Case I---the average received
interference power constraint at each PU and Case II---the total average received interference power constraint at all PUs\footnote{Notably, multiple single-antenna PUs exist. These PUs can also be considered a single equivalent PU with multiple receive antennas.}. These cases are respectively given by
\begin{subequations} \label{eq:the total received interference constrain}
\begin{align}
  &\mbox{Case I (Per PU power constraint):~~~~} \Ex \left\{ \qf_l^H \qG \qG^H \qf_l \right\} \leq P_l, \mbox{~~for}~~ l = 1,\ldots,L, \label{eq:the total received interference constrain 1}\\
  &\mbox{Case II (Sum power constraint):~~~~}\Ex \left\{ \tr\left( \qF \qG \qG^H \qF^H \right) \right\} \leq P_{\rm all},\label{eq:the total received interference constrain 2}
\end{align}
\end{subequations}
where $P_l > 0$ denotes the interference power threshold of ${\sf PU}_l$, and $P_{\rm all} > 0$ represents the total interference power threshold of all PUs.
We then set $P_l = \theta_l P_T$ and $P_{\rm all} = \theta_{\rm all} P_T$ with $\theta_l, \theta_{\rm all}$ being positive scalar parameters to make a
connection with the transmit power. Although we only consider equal power allocation for simplicity in this paper, our framework can be easily extended to
arbitrary power allocation by replacing $\qG$ with $\qG\qP^{\frac{1}{2}}$, where $\qP = \diag(p_1,\ldots, p_K)$ with $p_k\geq 0$ being the signal power of
${\sf SU}_k$ (see \cite{Wagner-12IT,Muharar-13TCom} for a similar application).

Next, to incorporate path loss and other large-scale fading effects, we model the channel vectors by
\begin{align}
 \qh_k^H & = \sqrt{r_{1,k}} \, \tqh_k^H \mbox{~~and}~~ \qf_l^H = \sqrt{r_{2,l}} \,  \tqf_l^H, 
\end{align}
where $\tqh_k^H$ and $\tqf_l^H$ are the small-scale (or fast) fading vectors, and $r_{1,k}$ and $r_{2,l}$ denote the large-scale fading coefficients (or
channel path gains), including the geometric attenuation and shadow effect. Using the above notations, the concerned channel matrices can be rewritten as
\begin{align} \label{eq:def_tilde_HandF}
 \qH & = \qR_{1}^{\frac{1}{2}} \tqH \mbox{~~and}~~\qF = \qR_{2}^{\frac{1}{2}} \tqF, 
\end{align}
where $\tqH \equiv [\frac{1}{\sqrt{N}} \tilde{h}_{ij}] \in \bbC^{K \times N}$ and $\tqF \equiv [\frac{1}{\sqrt{N}} \tilde{f}_{ij}] \in \bbC^{L \times N}$ consist
of the random components of the channel in which $\tilde{h}_{ij}$'s and $\tilde{f}_{ij}$'s are i.i.d.~complex random variables with zero mean and unit variance, respectively, and $\qR_1 \in \bbC^{K \times K}$ and $\qR_2 \in \bbC^{L \times L}$ are diagonal matrices whose diagonal elements are given by $[\qR_1]_{kk} = r_{1,k}$ and $[\qR_2]_{ll} = r_{2,l}$, respectively. In line with \cite{HeYY-13TWC,HeYY-13ICASSP}, we assume that $\qH$ is perfectly known to ${\sf BS}$ in this paper. Since ${\sf BS}$ needs to predict the interference power in \eqref{eq:the total received interference constrain}, we further assume that perfect knowledge of $\qF$ is available at ${\sf BS}$ \cite{LZhang-09TWC,HeYY-13TWC,HeYY-13ICASSP}. To acquire perfect channel state information (CSI) for $\qH$ and $\qF$, transmission protocols need to incorporate certain cooperation among the PUs, the SUs, and ${\sf BS}$ \cite{LZhang-09TWC}. Further research can focus on the case with imperfect CSI or estimation of channel \cite{Dai-13JSAC,Gao-14CL}.

In the downlink CR network \eqref{eq:the concatenated received signal vector}, we consider the RZF precoding because this precoding's relatively low complexity compared with dirty paper coding \cite{Joham-02ISSSTA,Peel-05Tcom,Wagner-12IT,ZhangJun-13TWC}. However, a direct application of the conventional RZF to the secondary
BS will result in a very inefficient transmission because a large power back-off at the secondary BS is required to satisfy the interference power constraint \eqref{eq:the total received interference constrain}. Therefore, following \cite{HeYY-13TWC,HeYY-13ICASSP}, we adopt the RZF precoding based on the \emph{partially-projected} channel matrix
\begin{equation} \label{eq:checkH}
    \check{\qH} =\qH(\qI_N-\beta\qW^H\qW),
\end{equation}
where $\qW \triangleq (\qF\qF^H)^{-\frac{1}{2}} \qF \in \bbC^{L \times N}$, and $\beta \in [0,1]$ is the projection control parameter. Note that the projected
channel matrix $\check{\qH}$ is obtained by \emph{partially} projecting $\qH$ into the null space of $\qF$. Specifically, the RZF precoding matrix is given by
\begin{equation}\label{eq:the RZF precoding}
 \qG =  \xi \left( \check{\qH}^H \check{\qH} + \alpha \qI_N \right)^{-1} \check{\qH}^H,
\end{equation}
where $\xi$ is a normalization parameter that fulfills the BS transmit power constraint \eqref{eq:base station transmitted power constrain} and the
interference power constraint \eqref{eq:the total received interference constrain}, and $\alpha > 0$ represents the regularization parameter. We refer to this
precoding as PP-RZF precoding.

Before setting each of the parameters in \eqref{eq:the RZF precoding}, two special cases of the PP-RZF precoding are considered first. On the one hand, if
$\beta=0$ then $\qG$ degrades to the conventional RZF precoding. On the other hand, if $\beta=1$ then $\check{\qH}$ is completely orthogonal to $\qF$ and we have $\qF\check{\qH}^H = \qzero$, i.e., no interference signal from the secondary BS will leak to the PUs. Therefore, the interference power constraint (\ref{eq:the total
received interference constrain}) is naturally guaranteed. Furthermore, the amount of the interference to the PUs decreases as the projection control
parameter increases.

Now we return to the setting of the normalization parameter in \eqref{eq:the RZF precoding}. Considering Case I, from \eqref{eq:base station transmitted power
constrain} and \eqref{eq:the total received interference constrain 1}, we have
\begin{subequations} \label{eq:xi_i}
\begin{align}
    \xi^2 \leq & \xi_0^2 \triangleq \frac{P_T}{ \Ex \left\{\frac{1}{N} \tr\left( \left( \check{\qH}^H \check{\qH} + \alpha \qI_N \right)^{-1} \check{\qH}^H \check{\qH} \left( \check{\qH}^H \check{\qH} + \alpha \qI_N \right)^{-1} \right) \right\} }, \\
    \xi^2 \leq & \xi_l^2 \triangleq \frac{\theta_l P_T}{ \Ex \left\{ \qf_l^H \left( \check{\qH}^H \check{\qH} + \alpha \qI_N \right)^{-1} \check{\qH}^H \check{\qH} \left( \check{\qH}^H \check{\qH} + \alpha \qI_N \right)^{-1} \qf_l \right\} }, \mbox{~~for}~ l = 1,\ldots,L.
\end{align}
\end{subequations}
To satisfy \eqref{eq:base station transmitted power constrain} and \eqref{eq:the total received interference constrain 1} simultaneously, we set $\xi^2 = \min \{\xi_0^2, \xi_l^2, l = 1,\ldots,L\}$. Then, the SINR of secondary user ${\sf SU}_k$ is given by
\begin{align}\label{eq:SINR}
   \gamma_{k} & = \frac{ \left| \qh_k^H \left( \check{\qH}^H \check{\qH} + \alpha \qI_N \right)^{-1} \check{\qh}_k \right|^2} { \qh_k^H \left( \check{\qH}^H \check{\qH} + \alpha \qI_N \right)^{-1} \check{\qH}_{[k]}^H \check{\qH}_{[k]} \left( \check{\qH}^H \check{\qH} + \alpha \qI_N \right)^{-1} \qh_k + \frac{\sigma^2}{\xi^2} } \nonumber\\
   &= \frac{ \rho \left| \qh_k^H \left( \check{\qH}^H \check{\qH} + \alpha \qI_N \right)^{-1} \check{\qh}_k \right|^2} { \rho \qh_k^H \left( \check{\qH}^H \check{\qH} + \alpha \qI_N \right)^{-1} \check{\qH}_{[k]}^H \check{\qH}_{[k]} \left( \check{\qH}^H \check{\qH} + \alpha \qI_N \right)^{-1} \qh_k + \nu},
\end{align}
where $\check{\qH}_{[k]} \triangleq [ \check{\qh}_1, \ldots, \check{\qh}_{k-1}, \check{\qh}_{k+1}, \ldots, \check{\qh}_K ]^H \in \bbC^{(K-1)\times N}$,
$\check{\qh}_k \triangleq (\qI_N-\beta\qW^H\qW)\qh_k$, $\rho \triangleq P_T/\sigma^2$, and
\begin{align}
 \nu \triangleq \frac{P_T}{\xi^2} =  \max &\bigg\{ \Ex \left\{ \frac{1}{N} \tr\left( \left( \check{\qH}^H \check{\qH} + \alpha \qI_N \right)^{-1} \check{\qH}^H \check{\qH} \left( \check{\qH}^H \check{\qH} + \alpha \qI_N \right)^{-1} \right) \right\} ,  \nonumber\\
   &~~~~\frac{1}{\theta_l} \Ex \left\{ \qf_l^H \left( \check{\qH}^H \check{\qH} + \alpha \qI_N \right)^{-1} \check{\qH}^H \check{\qH} \left( \check{\qH}^H \check{\qH} + \alpha \qI_N \right)^{-1} \qf_l \right\}, l = 1,\ldots,L \bigg\}. \label{eq:nu}
\end{align}
Here, the equality of (\ref{eq:nu}) follows from (\ref{eq:xi_i}). For Case II, we have
\begin{align}
 \nu = \max &\bigg\{ \Ex \left\{ \frac{1}{N} \tr\left( \left( \check{\qH}^H \check{\qH} + \alpha \qI_N \right)^{-1} \check{\qH}^H \check{\qH} \left( \check{\qH}^H \check{\qH} + \alpha \qI_N \right)^{-1} \right) \right\} ,  \nonumber\\
   &~~~~\frac{1}{\theta_{\rm all}} \Ex \left\{ \tr\left( \qF \left( \check{\qH}^H \check{\qH} + \alpha \qI_N \right)^{-1} \check{\qH}^H \check{\qH} \left( \check{\qH}^H \check{\qH} + \alpha \qI_N \right)^{-1} \qF^H \right) \right\} \bigg\}. \label{eq:nu2}
\end{align}
Consequently, under the assumption of perfect CSI at both transmitter and receivers, the ergodic sum-rate of the CR network with Gaussian signaling can be defined as
\begin{align}\label{eq:the ergodic sum-rate}
   R_{\rm{sum}} \triangleq \sum_{k=1}^K \Ex \left\{\log \left( 1 + \gamma_{k} \right)\right\}.
\end{align}
Note that $\gamma_k$ in the ergodic sum-rate is subject to the BS transmit power constraint in (3) and the interference power constraint (to the primary users) in (4).

\subsection{Problem Formulation}

The SINR $\gamma_{k}$ in \eqref{eq:SINR} is a function of the regularization parameter $\alpha$ and the projection control parameter $\beta$. In the
literature, adopting incorrect regularization parameter would degrade performance significantly \cite{Peel-05Tcom,Wagner-12IT,ZhangJun-13TWC}. In light of the discussion in the previous subsection, one can realize that a proper projection control parameter can assist in decreasing the interference to the
PUs. As a result, using the PP-RZF precoding effectively requires obtaining appropriate values of $\alpha$ and $\beta$ to optimize certain
performance metrics. In this paper, we are interested in finding $(\alpha,\beta)$, which maximizes the ergodic sum-rate (\ref{eq:the ergodic sum-rate}).
Formally, we have
\begin{align}\label{eq:optimal ergodic sum-rate}
   \left\{\alpha^{\rm opt}, \beta^{\rm opt} \right\} = \argmax_{\alpha > 0, 1 \geq \beta \geq 0} & ~R_{\rm{sum}}.
\end{align}

The above problem does not admit a simple closed-form solution and the solution must be computed via a two-dimensional line search. Monte-Carlo
averaging over the channels is required to evaluate the ergodic sum-rate \eqref{eq:the ergodic sum-rate} for each choice of $\alpha$ and $\beta$, which,
unfortunately, makes the overall computational complexity prohibitive. This drawback hinders the development of the PP-RZF precoding. To address this problem,
we resort to an asymptotic expression of \eqref{eq:the ergodic sum-rate} in the large-system regime in the next section.

\section{Performance Analysis of Large Systems}

This section presents the main results of the paper. First, we derive deterministic equivalents for the SINR $\gamma_{k}$ and the ergodic sum-rate
$R_{\rm{sum}}$ in a large-system regime. Then, we identify the asymptotically optimal regularization parameter and the asymptotically optimal projection control parameter
to achieve the optimal deterministic equivalent for the ergodic sum-rate.

\subsection{Deterministic Equivalents for the SINR and the Ergodic Sum-Rate}

We present a deterministic equivalent for the SINR $\gamma_{k}$ by considering the
large-system regime, where $N$, $K$, and $L$ approach infinity, whereas
\begin{equation*}
c_1 = \frac{N}{K} \mbox{~~and~~} c_2 = \frac{L}{N}
\end{equation*}
are fixed ratios, such that $0 < \lim \inf_N c_1 \leq \lim
\sup_N c_1 < \infty, 0 < \lim \inf_N c_2 \leq \lim \sup_N c_2 \leq 1$. For brevity, we simply use $\calN \rightarrow \infty$ to represent the quantity in such
limit. In addition, we impose the assumptions below in our derivations.
\begin{assumption} \label{Assum: 1}
For the channel matrices $\qH$ and $\qG$ in \eqref{eq:def_tilde_HandF}, we have the following hypotheses:
\begin{itemize}
\item[1)] $\tqH = [\frac{1}{\sqrt{N}} \tilde{h}_{ij}] \in \bbC^{K \times N}$, where $\tilde{h}_{ij}$'s
are i.i.d. standard Gaussian.

\item[2)] $\tqF = [\frac{1}{\sqrt{N}} \tilde{f}_{ij}] \in \bbC^{L \times N}$, where $\tilde{f}_{ij}$'s have the same statistical
properties as $\tilde{h}_{ij}$'s.

\item[3)] $\qR_1 = \diag(r_{1,1},\ldots, r_{1,K}) \in \bbC^{K \times K}$ and $\qR_2 = \diag(r_{2,1},\ldots, r_{2,L}) \in \bbC^{L \times L}$ are diagonal matrices with uniformly bounded spectral norm\footnote{\cite{Horn-90BOOK}: The spectral norm $|\!|\!|\bullet|\!|\!|_2$ is defined on $\bbC^{n\times n}$ by $|\!|\!|\qA|\!|\!|_2 \equiv \max \{\sqrt{\lambda}: \lambda \mbox{~~is an eigenvalue of~~} \qA^*\qA\}$.} with respect to $K$ and $L$, respectively.
\end{itemize}
\end{assumption}

Based on the definition of $\qW$ in \eqref{eq:checkH}, $\qW^H\qW = \qF^H (\qF\qF^H)^{-1}\qF = \tqF^H (\tqF\tqF^H)^{-1}\tqF$ $= \tqW^H\tqW$, where
$\tqW \triangleq (\tqF\tqF^H)^{-\frac{1}{2}}\tqF$. Therefore, $\tqW$ is $L \leq N$ rows of an $N \times N$ Haar-distributed unitary random matrix \cite[Definition 4.6]{Couillet-11BOOK}. The partially-projected channel matrix $\check{\qH}$ is clearly composed of the product of two different types of independent
random matrices. Owing to recent advancements in large dimensional RMT \cite{Couillet-11BOOK}, we arrive at the following theorem, and the
details are given in Appendix A.

\begin{theorem}\label{Th: 2}
Under Assumption \ref{Assum: 1}, in Case I (per PU power constraint), as $\calN \rightarrow \infty$, we have $\gamma_k - \overline{\gamma}_k  \xrightarrow{a.s.} 0$, for $k = 1,\dots,K$,
where
\begin{align}
 \overline{\gamma}_{k} &= \frac{ \rho \overline{a}_k^2} { \rho \overline{b}_k + \overline{\nu}},  \label{eq:gamma1 deterministic equivalent}
\end{align}
with
\begin{subequations} \label{eq:u_all}
\begin{align}
 \overline{a}_k&= \frac{r_{1,k} (t_1 + t_2(1-\beta)) }{\alpha+r_{1,k}(t_1 + t_2(1-\beta)^2)}, \label{eq:Dea}  \\
 \overline{b}_k &= r_{1,k} \left(  \frac{\left(1-\overline{a}_k\right)^2 t_1}{1+e} +  \frac{\left(1-(1-\beta)\overline{a}_k\right)^2 (1-\beta)^2 t_2}{1+(1-\beta)^2 e} \right)\frac{\partial e}{\partial \alpha} , \label{eq:Deb}\\
 \overline{\nu} & = \max \bigg\{ \left(\frac{t_1}{1+e}+\frac{(1-\beta)^2 t_2}{1+(1-\beta)^2 e}\right)\frac{\partial e}{\partial \alpha}, \frac{r_{2,l}}{\theta_l c_2} \frac{(1-\beta)^2 t_2}{1+(1-\beta)^2 e} \frac{\partial e}{\partial \alpha}, l = 1,\ldots,L  \bigg\},  \label{eq:Dev} \\
 \frac{\partial e}{\partial \alpha} &= \frac{\frac{1}{N} \tr \qR_1 \left(\alpha \qI_K + \left(t_1 + t_2 (1-\beta)^2\right) \qR_1 \right)^{-2}}{1 - \left(\frac{t_1}{1+e}+\frac{(1-\beta)^4 t_2}{1+(1-\beta)^2 e}\right) \frac{1}{N} \tr \left( \qR_1 \left(\alpha \qI_K + \left(t_1 + t_2 (1-\beta)^2\right) \qR_1 \right)^{-1}\right)^2}, \label{eq:partial e} \\
 t_1 &=  \frac{1-c_2}{1+e}, 
  ~~~~~~~~t_2 =  \frac{c_2}{1+(1-\beta)^2 e}, \label{eq:t1_f}
\end{align}
\end{subequations}
and $e$ is given as the unique solution to the fixed point equation
\begin{align}
  e =& \frac{1}{N} \tr \qR_1 \left(\alpha \qI_K + \left(t_1 + t_2 (1-\beta)^2\right) \qR_1 \right)^{-1}. \label{eq:e1}
\end{align}
Meanwhile, in Case II (sum power constraint), all asymptotic expressions remain, except for $\overline{\nu}$, which should be changed to
\begin{align}
 \overline{\nu} & = \max \bigg\{ \left(\frac{t_1}{1+e}+\frac{(1-\beta)^2 t_2}{1+(1-\beta)^2 e}\right)\frac{\partial e}{\partial \alpha}, \frac{\tr \qR_2}{\theta_{\rm all} c_2} \frac{(1-\beta)^2 t_2}{1+(1-\beta)^2 e} \frac{\partial e}{\partial \alpha} \bigg\}. \label{eq:Dev2}
\end{align}
\hfill $\blacksquare$
\end{theorem}

An intuitive application of Theorem \ref{Th: 2} is that $\gamma_{k}$ can be approximated by its deterministic equivalent\footnote{\cite[Definition 6.1]{Couillet-11BOOK} (also see \cite{Hachem-07AAP}): Consider a series of Hermitian random matrices $\qB_1, \,\qB_2, \, \ldots,$ with $\qB_N \in \bbC^{N \times N}$ and a series $f_1,f_2,\ldots$ of functionals of $1\times 1, 2\times 2, \ldots$ matrices. A deterministic equivalent of $\qB_N$ for functional $f_N$ is a series $\qB_1^{\circ}, \, \qB_2^{\circ}, \, \ldots$, where $\qB_N^{\circ}\in \bbC^{N\times N}$, of deterministic matrices, such that $\lim_{N\rightarrow\infty} f_N (\qB_N)-f_N(\qB_N^{\circ}) \rightarrow 0$. In this case, the convergence often be with probability one. Similarly, we term $g_N \triangleq f_N(\qB_N^{\circ})$ the deterministic equivalent of $f_N (\qB_N)$, that is, the deterministic series $g_1, g_2, \ldots$, such that $f_N (\qB_N) - g_N \rightarrow 0$ in some sense. \\
Note that the deterministic equivalent of the Hermitian random matrix $\qB_N$ is a \emph{deterministic} and a \emph{finite dimensional} matrix $\qB_N^{\circ}$. In addition, the deterministic equivalent of $f_N (\qB_N)$ is $g_N \triangleq f_N(\qB_N^{\circ})$, which is a function of $\qB_N^{\circ}$.} $\overline{\gamma}_{k}$, which can be determined based only on statistical channel knowledge, that is, $\qR_1$, $\qR_2$, and $\sigma^2$. Note that, according to the definition of the deterministic equivalent (see footnote 3), in the expression of the deterministic equivalent $\overline{\gamma}_{k}$, the parameters $N$, $K$, $L$, as well as the matrix dimensions of $\qR_1$ and $\qR_2$, are \emph{finite}.
Combining Theorem 1 with the continuous mapping theorem\footnote{\cite[Theorem 25.7-Corollary 2]{Billingsley-11BOOK}: If $x_n \xrightarrow{a.s.} a$ and $h$ is continuous function at $a$, then $h(x_n) \xrightarrow{a.s.} h(a)$.}, we have $\log \left( 1 + \gamma_{k} \right) - \log \left( 1 + \overline{\gamma}_{k} \right)\xrightarrow{a.s.} 0 $. An approximation $\overline{R}_{\rm{sum}}$ of the ergodic sum-rate $R_{\rm{sum}}$ in \eqref{eq:the ergodic sum-rate} is obtained by replacing the instantaneous SINR $\gamma_{k}$ with its large system approximation $\overline{\gamma}_{k}$, that is,
\begin{equation} \label{eq:deterministic equivalent sum-rate}
 \overline{R}_{\rm{sum}} =\sum_{k=1}^K \log \left( 1 + \overline{\gamma}_{k} \right).
\end{equation}
Therefore, when $\calN \rightarrow\infty, \frac{1}{K} \left( R_{\rm{sum}} - \overline{R}_{\rm{sum}} \right) \xrightarrow{a.s.} 0$
holds true almost surely.

To facilitate our understanding of Theorem \ref{Th: 2}, we look at it from the two special cases as follows:
\begin{enumerate}
\item In Theorem \ref{Th: 2}, we introduce the two variables $t_1$ and $t_2$ to reflect the effects of the projection control parameter $\beta$.
If $\beta = 1$, from \eqref{eq:u_all}, then the deterministic equivalent $\overline{\gamma}_{k}$ does not depend on $t_2$. Substituting $\beta = 1$ into
\eqref{eq:gamma1 deterministic equivalent} and letting $\qR_1 = \qI_K$, we have
\begin{align} \label{eq:gammabeta1}
 \overline{\gamma}_{k} &= \frac{\rho\Big(c_1(1-c_2)(1+\zeta(\mu,\eta,\alpha))^2-\zeta(\mu,\eta,\alpha)^2 \Big)}{\rho + \Big(1 +\zeta(\mu,\eta,\alpha)\Big)^2},
\end{align}
where $\zeta(\mu,\eta,\alpha) \triangleq t_1/\alpha$, $\mu \triangleq 1-c_2$, and $\eta \triangleq 1/c_1$. Combining \eqref{eq:t1_f} and \eqref{eq:e1}, we
obtain
\begin{equation} \label{eq:t1_SpecCase}
    \zeta(\mu,\eta,\alpha) \triangleq \frac{t_1}{\alpha}
    = \frac{1}{2} \left( \frac{\mu-\eta}{\alpha} - 1 + \sqrt{\frac{(\mu-\eta)^2}{\alpha^2} + \frac{2(\mu+\eta)}{\alpha} +1 } \, \right).
\end{equation}
Before providing an observation based on the above, we briefly review a well-known result from the large dimensional RMT. First, we consider the
definition of $\qH$ from (\ref{eq:def_tilde_HandF}). If $\qR_1 = \qI_K$, the entries of the $K \times N$ matrix $\qH$ are zero mean i.i.d. with variance $1/N$.
Following \cite[Chapter 3]{Couillet-11BOOK}, we see that as $N, \, K \rightarrow \infty$ with $N/K \rightarrow c_1$, $\qh_k^H \left( \qH^H \qH + \alpha \qI_N
\right)^{-1} \qh_k$ converges almost surely to
\begin{equation} \label{eq:t0_def}
   \int_{a}^{b} \frac{1}{\lambda+\alpha} f(\lambda) \,d\lambda,
\end{equation}
where
\begin{equation}
  f(\lambda) = \left(1- \eta \right)^{+}\delta(\lambda) + \frac{\sqrt{(\lambda-a)^{+}(b-\lambda)^{+}}}{2\pi \lambda}
\end{equation}
with $(x)^{+} \triangleq \max\{x,\,0\}$, $a \triangleq (1- \sqrt{\eta})^2$, and $b \triangleq (1+ \sqrt{\eta})^2$. In fact, $f(u)$ is the limiting empirical
distribution of the eigenvalues of $\qH^H \qH$ and is known as the Mar\v{c}cenko-Pastur law \cite{Mar-67}. The integral of (\ref{eq:t0_def}) can be evaluated
in closed form
\begin{equation} \label{eq:t0}
   \frac{1}{2} \left( \frac{1-\eta}{\alpha} - 1 + \sqrt{\frac{(1-\eta)^2}{\alpha^2} + \frac{2(1+\eta)}{\alpha} +1 } \, \right).
\end{equation}
Note that (\ref{eq:t1_SpecCase}) is equal to (\ref{eq:t0}) when $\mu$ is replaced with $1$, i.e., \eqref{eq:t0} is equal to $\zeta(1,\eta,\alpha)$.
In fact, following the similar derivations of Theorem \ref{Th: 2}, we can show that \eqref{eq:gammabeta1} and the SINR of the conventional RZF precoding share
the same formulation by replacing $\zeta(\mu,\eta,\alpha)$ in \eqref{eq:gammabeta1} with $\zeta(1,\eta,\alpha)$. Substituting the definitions of $c_1,
c_2$ into $\mu$ and $\eta$, we have $\mu-\eta = 1-c_2-1/c_1 = (N-(L+K))/N$. Comparing this value with $1-\eta = (N-K)/N$ in \eqref{eq:t0}, we thus
conclude that if $\beta =1$, the SINR of the PP-RZF precoding is \emph{similar}\footnote{Notably, when $\beta =1$, the SINRs of the PP-RZF precoding and the conventional RZF precoding are similar but \emph{not} identical because $\zeta(\mu,\eta,\alpha)$ is replaced with $\zeta(1,\eta,\alpha)$.} to that of the conventional RZF precoding but with an increase in the number of active users from $K$ to $K+L$. Hence, the degrees of freedom of the PP-RZF precoding is reduced to $N-(K+L)$ because the additional $L$ degrees of freedom are used to suppress interference to the PUs.

\item For another extreme case with $\beta = 0$ in Theorem \ref{Th: 2}, $t_1+t_2=\frac{1}{1+e}$. Letting $\qR_1 = \qI_K$,
we obtain $\frac{1}{\alpha(1+e)} = \zeta(1,\eta,\alpha) $, such that
\begin{align} \label{eq:gammabeta0}
 \overline{\gamma}_{k} &= \frac{\rho\Big(c_1(1+\zeta(1,\eta,\alpha))^2-\zeta(1,\eta,\alpha)^2 \Big)}{\rho + \nu_0\Big(1+\zeta(1,\eta,\alpha)\Big)^2},
\end{align}
where $ \nu_0 = \max \{1, r_{2,l}/\theta_l, l=1,\ldots,L\}$ is for Case I and $ \nu_0 = \max \{1, \tr\qR_2/\theta_{\rm all}\}$ is for Case II. The received
interference power constraint at the PUs (\ref{eq:the total received interference constrain}) can be controlled only through $\nu_0$, where $\beta$ is
not involved in $\nu_0$. Therefore, the SINR $\overline{\gamma}_{k}$ is significantly degraded if the channel path gains between the BS and the PUs (that is,
$r_{2,l}$'s) are strong. However, if the channel path gains between the BS and the PUs are weak, then $\nu_0 = 1$ and $\overline{\gamma}_{k}$ behave in a manner similar to but \emph{not} identical to that of the conventional RZF precoding because $c_1$ is replaced with $c_1(1-c_2)$.

Comparing (\ref{eq:gammabeta0}) for $\beta = 0$ with (\ref{eq:gammabeta1}) for $\beta = 1$ obtains notable results. First, we note that
(\ref{eq:gammabeta1}) and (\ref{eq:gammabeta0}) share a similar formulation, except the additional $\nu_0$ appears at the denominator of
(\ref{eq:gammabeta0}). When $\beta = 1$, the secondary BS yields zero interference on the PUs, such that the interference power constraint in
(\ref{eq:the total received interference constrain}) is always inactive. Therefore, no additional parameter $\nu_0$ is required to reflect the received
interference power constraint at the PUs. Although $\nu_0 \geq 1$, the SINR performance of the PP-RZF precoding with $\beta = 1$ is not implied to
be always better than that with $\beta = 0$. An additional note should be given on $\zeta(\cdot,\eta,\alpha) $, where the argument $\cdot$ is $\mu$ for
$\beta = 1$ and $1$ for $\beta = 0$. The parameter $\mu = (N-L)/N$ for $\beta = 1$ implies that the additional $L$ degrees of freedom is used to suppress
interference to the PUs. Consequently, if the channel path gains between the BS and the PUs are weak, the SINR performance of the PP-RZF precoding with
$\beta = 1$ shall not be better than that with $\beta = 0$. Thus, we infer that the projection control parameter should be decreased if the received
interference power constraint at the PUs is relaxed.

Finally, we note that $\zeta(1,\eta,\alpha)$ agrees with $z(r,\alpha_0)$ in \cite[Theorem 1]{HeYY-13TWC,HeYY-13ICASSP}. As a result,
(\ref{eq:gammabeta0}) is identical to the deterministic equivalent for the SINR obtained in \cite[Theorem 1]{HeYY-13TWC,HeYY-13ICASSP}, where the PP-RZF precoding
with a \emph{single} PU is considered. The deterministic equivalent for the SINR in \cite[Theorem 1]{HeYY-13TWC,HeYY-13ICASSP} is clearly a special case
of \eqref{eq:gamma1 deterministic equivalent} with $\beta = 0$ even though the case of $\beta \neq 0$ is considered in \cite[Theorem 1]{HeYY-13TWC,HeYY-13ICASSP} because a single PU results only in one-dimensional perturbation, and the effect of such perturbation \emph{vanishes} in a large system. Even if the number of PUs $L$ is finite and only $N$ becomes large, the effect of $\beta$ vanishes. The lack of a relation between $\beta$ and the SINRs will result in a bias when the number of antennas at the BS is not so large. However, our analytical results show the effect of $\beta$ by assuming that $N$, $K$, and $L$ are large, whereas  $c_1 = N/K$ and $c_2 = L/N$ are fixed ratios. Thus, our results are clearly more general than those in \cite{HeYY-13TWC,HeYY-13ICASSP}.
\end{enumerate}

\begin{corollary}\label{Co: Th1}
In addition to the assumptions of Theorem \ref{Th: 2}, we suppose further that $c_2 = 1$ (that is, $N = L$), $\qR_1 = r_1 \qI_K$, and $\beta \in [0,1)$.
Then, as $\calN \rightarrow \infty$, we have $\gamma_{k} - \overline{\gamma} \xrightarrow{a.s.} 0$ for $k = 1,\dots,K$, where
\begin{align}
 \overline{\gamma} &= \frac{ \rho \left( c_1 r_1^2-(c_1 \alpha e - r_1)^2 \right)} { \rho (c_1 \alpha e)^2 + \nu_0},  \label{eq:Degamma1}
\end{align}
and $e$ is given as an unique solution to the fixed point equation
\begin{align}
 e = \frac{r_1(1+e(1-\beta)^2)}{c_1 \alpha (1+e(1-\beta)^2) + c_1 r_1 (1-\beta)^2}, \nonumber
\end{align}
and $ \nu_0 = r_1 \max \{1, r_{2,l}/\theta_l, l=1,\ldots,L\}$ for Case I or $ \nu_0 = r_1 \max \{1, \tr\qR_2/\theta_{\rm all}\}$ for Case II.
\end{corollary}
\begin{proof}
By letting $c_2 = 1$ and $\qR_1 = r_1 \qI_K$, we immediately obtain the result from Theorems \ref{Th: 2} and \ref{Th: 1}.
\end{proof}

For a brief illustration, we consider only Case II of Corollary \ref{Co: Th1} because the same characteristics can be found in Case I. Given that
$\theta_{\rm all} = P_{\rm all}/P_T = P_{\rm all}/(\sigma^2 \rho)$, \eqref{eq:Degamma1} can be rewritten as
\begin{equation}\label{eq:gammaRemark}
 \overline{\gamma} = \left\{
\begin{aligned}
&\frac{   c_1 r_1^2-(c_1 \alpha e - r_1)^2 } { (c_1 \alpha e)^2 + 1/\rho}, & & 0<\frac{\rho \sigma^2\tr\qR_2}{P_{\rm all}} \leq 1; \\
&\frac{   c_1 r_1^2-(c_1 \alpha e - r_1)^2 } { (c_1 \alpha e)^2 + \sigma^2\tr\qR_2/P_{\rm all} }, & & 1 < \frac{\rho \sigma^2\tr\qR_2}{P_{\rm all}} .
\end{aligned}
\right.
\end{equation}
We can see that $\overline{\gamma}$ does not depend on the SNR $\rho$ when $1 < \rho \sigma^2\tr\qR_2/ P_{\rm all}$. In this case, the system performance is
interference-limited. Notably, the assumptions of $c_2 = 1$ and $\beta \neq 1$ are taken in Corollary \ref{Co: Th1}. In the case of $c_2 = 1$ and
$\beta = 1$, from \eqref{eq:Dea}, we have $\overline{a}_k = 0$ and consequently $\overline{\gamma}_{k}=0$, which implies a failure in the transmission. This result
is reasonable because when $c_2 = 1$, the dimension of the null space of $\qF$ is zero with probability one.\footnote{ If $N = L$, we have ${\sf
Rank}(\qI-\qF^H(\qF\qF^H)^{-1}\qF) = 0$ with probability one because from \cite[Theorem 1.1]{Rudelson-09Arxiv}, $\qF$ is a full rank square matrix with probability
one. } Therefore, the setting of $\beta = 1$ results in transmission failure, \emph{even} when the channel path gains between the BS and the PUs are weak. We thus show that a choice of appropriate $\beta$ significantly affects the successful operation of the CR network, which serves as
motivation for the remainder of this paper.

\subsection{Asymptotically Optimal Parameters}

Our numerical results confirm the high accuracy of the deterministic equivalent for the ergodic sum-rate $\overline{R}_{\rm{sum}}$ in the next section. Therefore, the
deterministic equivalent for the ergodic sum-rate can be used to determine the regularization parameter $\alpha$ and the projection control parameter $\beta$. By replacing
$R_{\rm{sum}}$ with $\overline{R}_{\rm{sum}}$ in \eqref{eq:optimal ergodic sum-rate}, we focus on this particular optimization to maximize the deterministic
equivalent for the ergodic sum-rate
\begin{align}\label{eq:optimal de sum-rate}
   \left\{\overline{\alpha}^{\rm opt}, \overline{\beta}^{\rm opt} \right\} = \argmax_{\alpha > 0, 1 \geq \beta \geq 0} & \overline{R}_{\rm{sum}}.
\end{align}
Similar to the problem in (\ref{eq:optimal ergodic sum-rate}), the asymptotically optimal solutions $\overline{\alpha}^{\rm opt}$ and $\overline{\beta}^{\rm
opt}$ do not permit closed-form solutions. However, the asymptotically optimal solution can be computed efficiently via the following methods without the need for Monte-Carlo averaging because $\overline{\gamma}_{k}$ is deterministic. First, given that $\beta$ is fixed, the optimal $\overline{\alpha}^{\rm opt}(\beta) := \argmax_{\alpha > 0}  \overline{R}_{\rm{sum}}(\beta)$ can be obtained efficiently via one-dimensional line search \cite{Wagner-12IT,ZhangJun-13TWC}, which performs the simple gradient method. The complexity in this part is linear. Then, we obtain the optimal $\overline{\beta}^{\rm opt} := \argmax_{0 \leq \beta \leq 1} \overline{R}_{\rm{sum}}(\overline{\alpha}^{\rm opt}(\beta),\beta)$ through the one-dimensional exhaustive search\footnote{Although the one-dimensional exhaustive search seems burdensome, the case in question here is easy because the search is only over a closed set $0 \leq \beta \leq 1$.}. Finally, the optimal parameters are given by $\{\overline{\alpha}^{\rm opt}(\overline{\beta}^{\rm opt}), \overline{\beta}^{\rm opt} \}$. For a special case, we obtain a condition of the optimal solutions in the following proposition:
\begin{proposition}\label{proposition: opt alphabeta}
Under the assumptions of Corollary \ref{Co: Th1}, the asymptotically optimal parameters $\overline{\alpha}^{\rm opt}$ and $\overline{\beta}^{\rm opt}$ satisfy the equation
\begin{align}
  \overline{\alpha}^{\rm opt} = \frac{\nu_0 \left(1-\overline{\beta}^{\rm opt}\right)^2}{\rho c_1 r_1}. \label{eq:optimal alphabeta}
\end{align}
where $\overline{\beta}^{\rm opt} \in [0,1)$.
\end{proposition}
\begin{proof}
By differentiating $\overline{R}_{\rm{sum}}$ with respect to $\alpha$ and $\beta$, we immediately obtain the result from Corollary \ref{Co: Th1}.
\end{proof}

From Proposition \ref{proposition: opt alphabeta}, we note that the number of asymptotically optimal solutions is infinite. All $\alpha$'s and $\beta$'s
that satisfy \eqref{eq:optimal alphabeta} are optimal. This condition will be confirmed in the next section.

Similar to \eqref{eq:gammaRemark}, we consider Case II for brief illustration. In this case, \eqref{eq:optimal alphabeta} can be rewritten as
\begin{equation}\label{eq:alphaRemark}
 \overline{\alpha}^{\rm opt} = \left\{
\begin{aligned}
&\frac{ (1-\overline{\beta}^{\rm opt})^2}{\rho c_1 r_1}, & & 0<\frac{\rho \sigma^2\tr\qR_2}{P_{\rm all}} \leq 1; \\
&\frac{ (1-\overline{\beta}^{\rm opt})^2 \sigma^2\tr\qR_2 }{ c_1 r_1 P_{\rm all}}, & & 1 < \frac{\rho \sigma^2\tr\qR_2}{P_{\rm all}} .
\end{aligned}
\right.
\end{equation}
From \eqref{eq:gammaRemark}, when $0< \rho \sigma^2\tr\qR_2/P_{\rm all} \leq 1$, the system performance is unaffected by the average received
interference power constraint. In this case, $\overline{\beta}^{\rm opt}$ is expected to be close to $0$ because the weak interference at all the PUs is negligible. This condition is combined with the first term of \eqref{eq:alphaRemark} to reveal that $\overline{\alpha}^{\rm opt}$ decreases with
increasing $\rho$, where $\rho = P_T/\sigma^2$ is the same as previously defined. However, when $\rho \sigma^2\tr\qR_2/P_{\rm all} >1$, the system
performance is limited by the average received interference power constraint. To decrease the interference, $\overline{\beta}^{\rm opt}$ is expected to be
close to $1$. Therefore, the second term of \eqref{eq:alphaRemark} reveals that $\alpha$ decreases to $0$ with an increase in $\overline{\beta}^{\rm opt}$.

We end this section by observing two additional extreme cases in Theorem \ref{Th: 2} for $\qR_1 = \qI_K$: If $\beta = 0$, by means of some
algebraic manipulations, we obtain $\overline{\alpha}^{\rm opt} = \nu_0/(c_1\rho)$. By contrast, if $\beta = 1$ and $c_2 \neq 1$, we obtain
$\overline{\alpha}^{\rm opt} = 1/(c_1 \rho)$. We find that the optimal regularization parameter tends to decrease monotonically with increasing $\rho$,
as expected. This characteristic is similar to that of the conventional RZF precoding in \cite{Peel-05Tcom,Nguyen-08GLCOM}, where $r_1= 1$ is assumed and
the asymptotically optimal regularization parameter $\overline{\alpha}^{\rm opt} = 1/(c_1 \rho )$ is derived.

\section{Simulations}\label{Se:simulations}

In this section, we conduct simulations to confirm our analytical results. First, we compare the analytical results \eqref{eq:deterministic equivalent
sum-rate} in Theorem \ref{Th: 2} and the Monte-Carlo simulation results \eqref{eq:the ergodic sum-rate} obtained from averaging over a large number of i.i.d. Rayleigh
fading channels. In the simulations, we set channel path gains $r_{1,k} = 1$ and $r_{2,l} = 0.6$ for all $k$ and $l$ and assume that $P_l = P$ for all $l$ in Case
I and $P_{\rm all} = LP$ in Case II. Several characteristics of Cases I and II are similar. Thus, without loss of generality, we provide the numerical results of Case I only.

\begin{figure}
\begin{center}
\resizebox{4.5 in}{!}{%
\includegraphics*{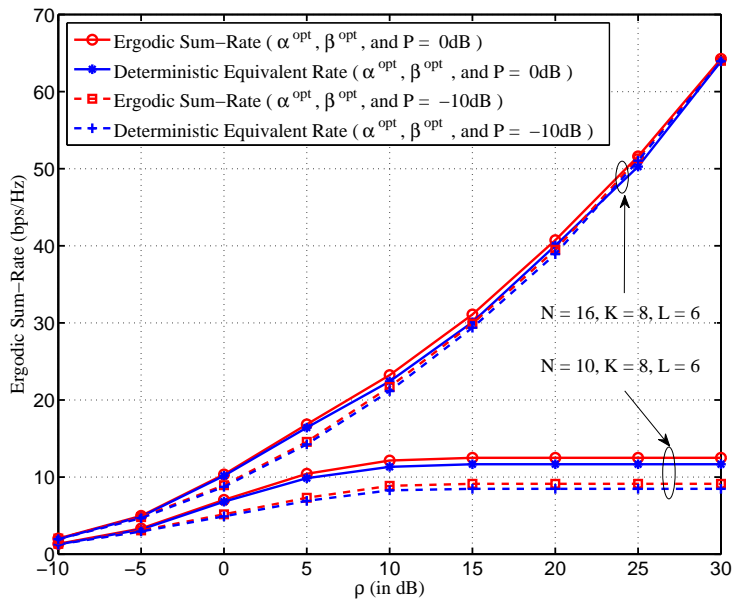} }%
\caption{Ergodic sum-rate and the deterministic equivalent results under different interference power threshold and and two different antenna configuration cases.}\label{fig:2}
\end{center}
\end{figure}

Fig. \ref{fig:2} compares the ergodic sum-rate and its deterministic equivalent result under different interference power thresholds $P\in \{-10 {\rm dB} ,
\, 0{\rm dB} \}$ and two different antenna configuration cases: $\{ N=10,~K=8,~L=6\}$ and $\{N=16,~K=8,~L=6\}$. In the simulation, $\{\alpha^{\rm opt}, \beta^{\rm
opt} \}$ is obtained by using the two-dimensional line search in \eqref{eq:optimal ergodic sum-rate}. We find that the deterministic
equivalent is accurate under various settings even for systems with a not-so-large number of antennas. In addition, Fig. \ref{fig:2} illustrates that for the case
with $\{ N=10,~K=8,~L=6\}$, the sum-rate of the SUs cannot increase linearly in SNR and becomes interference-limited because the sum-rate of the SUs is
easily restricted by the average received interference power at each PU, particularly when the number of active users is larger than the number of
antennas at the BS, that is, $ L+K \geq N$.

\begin{figure}
\begin{center}
\resizebox{4.5 in}{!}{%
\includegraphics*{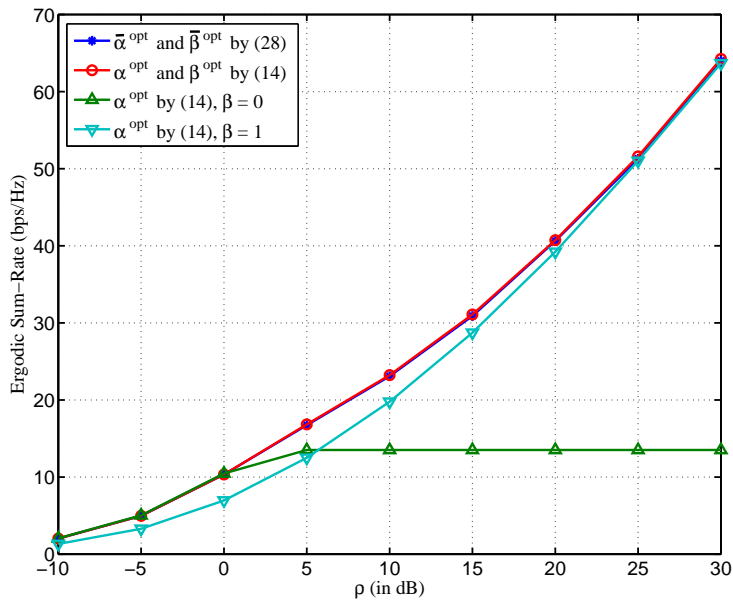} }%
\caption{Ergodic sum-rate results under various parameters with $P=0$dB and $\{N=16,~K=8,~L=6\}$.}\label{fig:3}
\end{center}
\end{figure}

In the above simulations, the best solutions of $\{\alpha^{\rm opt}, \beta^{\rm opt} \}$ are calculated by Monte-Carlo averaging over $10^4$ independent trials; doing so
which clearly results in a high computational cost. To confirm that the optimization based on the deterministic equivalent is not only more
computationally efficient but also near-optimal, we compare the ergodic sum-rate of the PP-RZF precoding with $P=0$dB and $\{ N=16,~K=8,~L=6\}$ in Fig.
\ref{fig:3} for the following four cases: 1) $\{\overline{\alpha}^{\rm opt}, \overline{\beta}^{\rm opt} \}$, 2) $\{\alpha^{\rm opt}, \beta^{\rm opt} \}$, 3)
$\{\alpha^{\rm opt}, \beta=0\}$, and 4) $\{\alpha^{\rm opt}, \beta=1\}$. The solution of $\{\overline{\alpha}^{\rm opt}, \overline{\beta}^{\rm opt} \}$ is obtained by
using the two-dimensional line search in \eqref{eq:optimal de sum-rate}. $\{\overline{\alpha}^{\rm opt}, \overline{\beta}^{\rm opt} \}$ provides
results that are indistinguishable from those achieved by $\{\alpha^{\rm opt}, \beta^{\rm opt} \}$, which demonstrates that the optimization based on the deterministic
equivalent is promising. Moreover, the performance is significantly improved if the PP-RZF precoding with an appropriate choice of $\{\alpha, \beta \}$
is employed. In the low-SNR regime, the optimal transmission becomes the conventional RZF precoding, whereas the optimal transmission is the PP-RZF
precoding with $\beta=1$ in the high-SNR regime.

\begin{figure}
\begin{center}
\resizebox{4.5 in}{!}{%
\includegraphics*{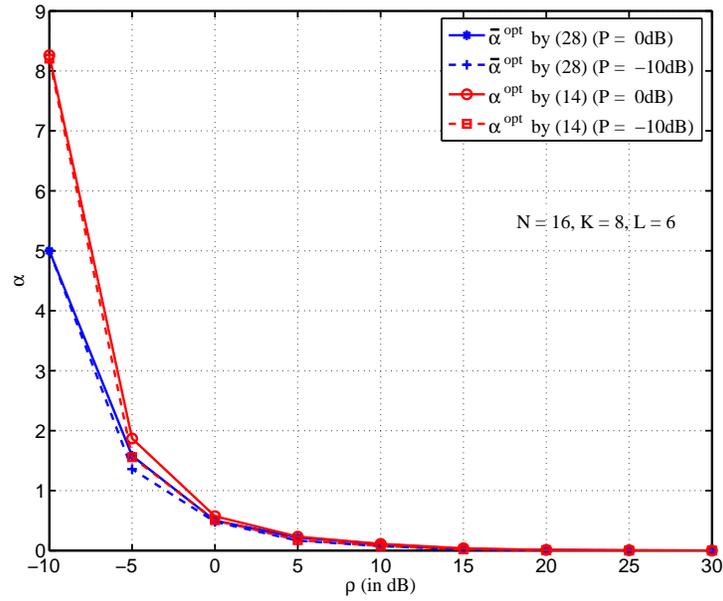} }%
\caption{Optimal $\alpha$ under different the interference power threshold for $\{N=16,~K=8,~L=6\}$.}\label{fig:4}
\end{center}
\end{figure}

\begin{figure}
\begin{center}
\resizebox{4.5 in}{!}{%
\includegraphics*{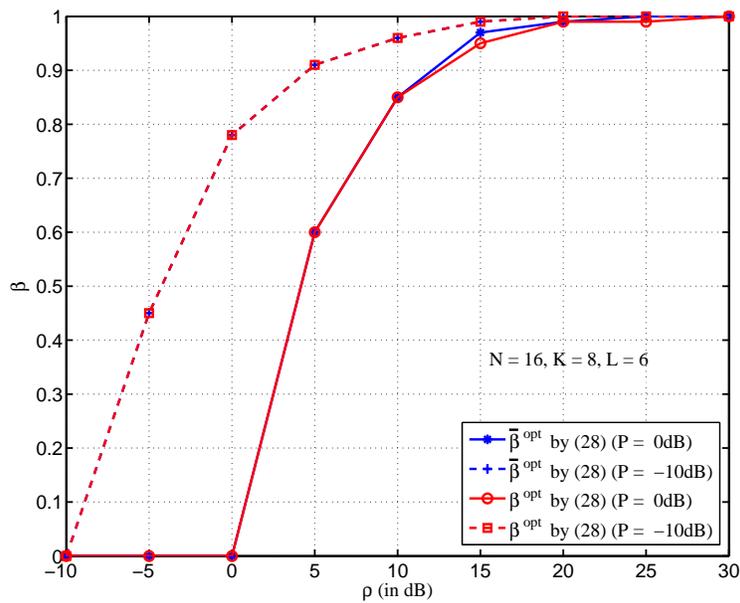} }%
\caption{Optimal $\beta$ under different interference power threshold for $\{N=16,~K=8,~L=6\}$.}\label{fig:5}
\end{center}
\end{figure}

To provide further results on the optimal solutions of $\{\alpha, \beta \}$, Figs. \ref{fig:4} and \ref{fig:5} show the values of $\{\overline{\alpha}^{\rm
opt}, \overline{\beta}^{\rm opt} \}$, $\{\alpha^{\rm opt}, \beta^{\rm opt} \}$ under various settings. We have observed that the optimal parameter
$\{\overline{\alpha}^{\rm opt}, \overline{\beta}^{\rm opt} \}$ based on the deterministic equivalent result is almost consistent with $\{\alpha^{\rm opt},
\beta^{\rm opt} \}$ based on the ergodic sum-rate. Moreover, we have observed that with increasing $\rho$, $\alpha^{\rm opt}$ (or $\overline{\alpha}^{\rm opt}$)
tends to monotonically decrease to $0$, whereas $\beta^{\rm opt}$ (or $\overline{\beta}^{\rm opt}$) tends to monotonically increase from $0$ to $1$. These
characteristics are expected based on the analysis in Section III.

Finally, we confirm the result in Proposition \ref{proposition: opt alphabeta}. Fig. \ref{fig:6} displays the ergodic sum-rate under various parameter settings
with $P=0$dB and $\{N=10,~K=8,~L=10\}$. We find that when $c_2=1$, the parameters that satisfy \eqref{eq:optimal alphabeta} can achieve the
asymptotically optimal sum-rate for any $\beta \in [0,1)$, such that infinitely many asymptotically optimal solutions exist.

\begin{figure}
\begin{center}
\resizebox{4.5 in}{!}{%
\includegraphics*{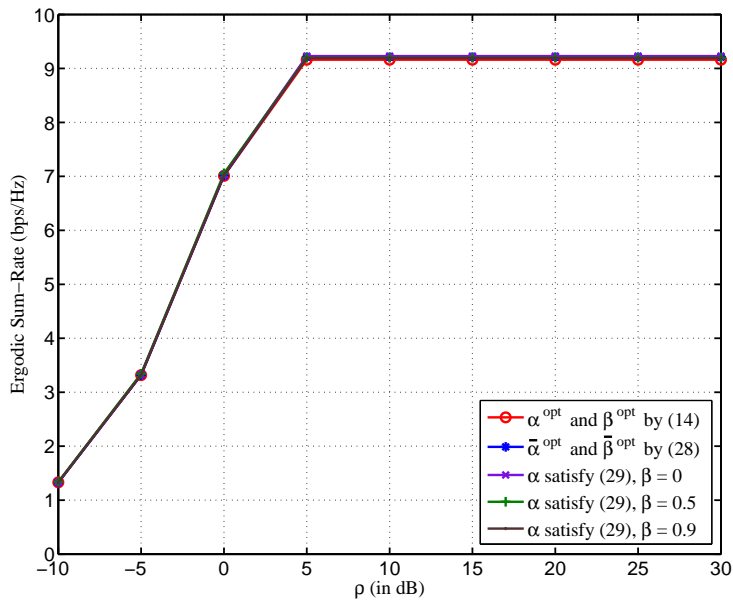} }%
\caption{Ergodic sum-rate results under various parameters for $P=0$dB and $\{N=10,~K=8,~L=10\}$.}\label{fig:6}
\end{center}
\end{figure}

\section{Conclusion}

By exploiting the recent advancements in large dimensional RMT, we investigated downlink multiuser CR networks that consist of multiple SUs and multiple PUs.
The deterministic equivalent of the ergodic sum-rate based on the PP-RZF precoding was derived. Numerical results revealed that the deterministic equivalent sum-rate provides reliable performance predictions even for systems with a not-so-large number of
antennas. We thus used the deterministic equivalent result to identify the asymptotically optimal regularization parameter and the asymptotically optimal projection
control parameter. In addition, we provided the condition that the regularization parameter and the projection control parameter are asymptotically optimal. Several
insights have been gained into the optimal PP-RZF precoding design. A natural extension of this is to consider the PP-RZF precoding under various scenarios, such as spatial correlations and imperfect CSI at the transmitter. However, such development is still ongoing because of mathematical difficulties.


\section*{Appendix A: Proof of Theorem \ref{Th: 2}} \label{Appendix:Th2proof}

To complete this proof, we first introduce the limiting distribution for a new class of random Hermitian matrix in Theorem \ref{Th: 1}. Such distribution serves as the mathematical basis for the latter derivation. We recall the definition of the Stieltjes transform (see, e.g., \cite{SilversteinBai-95}). For a Hermitian matrix $\qB_N \in \bbC^{N\times N}$, the Stieltjes transform of $\qB_N$, is defined as
\begin{equation}
 m_{\qB_N}(\alpha) =\frac{1}{N}\tr \left(\qB_N + \alpha \qI_N\right)^{-1} ~~ \mbox{for}~ \alpha \in {\bbR^+}.  \nonumber
\end{equation}
For ease of explanation, we also define the matrix product Stieltjes transform of $\qB_N$ as
\begin{equation}
 m_{\qB_N, \qQ}(\alpha)=\frac{1}{N}\tr \qQ \left(\qB_N + \alpha \qI_N\right)^{-1}, \nonumber
\end{equation}
where $\qQ$ is any matrix with bounded spectrum norm (with respect to $N$).

Notably, both $m_{\qB_N}(\alpha)$ and $m_{\qB_N, \qQ}(\alpha)$ are functions of $\alpha$, but for ease of notation, $\alpha$ is dropped. In addition, all the
subsequent approximations will be performed under the limit $\calN \rightarrow \infty$, and for ease of expression, $a \asymp b$ denotes that $a-b \xrightarrow{a.s.} 0$ as $\calN \rightarrow \infty$.

\begin{theorem}\label{Th: 1}
Consider an $N \times N$ matrix of the following form:
\begin{equation}\label{eq:B}
 \qB_N = \check{\qH}^H \check{\qH} = \big(\qI_N-\beta\tqW^H\tqW \big) \tqH^H \qR_{1} \tqH \big(\qI_N-\beta\tqW^H\tqW \big),
\end{equation}
where $\tqW$, $\tqH$, and $\qR_{1}$ follow the restrictions given by Assumption \ref{Assum: 1}. Then, as $\calN \rightarrow \infty$, we have
\begin{equation}
  m_{\qB_N, \qQ} \asymp  \frac{ t_1 + t_2}{\alpha}   \frac{1}{N} \tr \qQ, \label{eq:trQB}
\end{equation}
where $t_1 =  \frac{1-c_2}{1+e}$ and $t_2 = \frac{c_2}{1+e(1-\beta)^2}$ with $e$ being the unique solution to the fixed point equation
\begin{align}
  e =& \frac{1}{N} \tr \qR_1 \left(\alpha \qI_K + \left(t_1 + t_2 (1-\beta)^2\right) \qR_1 \right)^{-1}. \label{eq:e}
\end{align}
\end{theorem}
\begin{proof}
If $c_2 = 1$ (i.e., $N = L$), $\tqW^H \tqW = \qI_L$, the result is directly obtained by Lemma \ref{Lemma:TXXT} (see Appendix C).

We consider the case with $c_2 < 1$. Given that $m_{\qB_N, \qQ}$ is a function of two random matrices $\tqW$ and $\tqH$, we aim to derive an iterative deterministic equivalent \cite{Hoydis-11} of $m_{\qB_N, \qQ}$. In particular, we first find a function $\tilde{g}_N(\tqW,\alpha)$, such that $f_N((\tqH,\tqW),\alpha) \asymp \tilde{g}_N(\tqW,\alpha)$, where
$f_N((\tqH,\tqW),\alpha)\triangleq m_{\qB_N, \qQ}$, and $\tilde{g}_N(\tqW,\alpha)$ is a function of $\tqW$ and is independent of $\{\tqH\}_{N\geq 1}$. Notably,
$\tilde{g}_N(\tqW,\alpha)$ is a deterministic equivalent of $f_N((\tqH,\tqW),\alpha)$ with respect to random matrix sequences $\{\tqH\}_{N\geq 1}$. Second, we further find a function $g_N(\alpha)$, such that $\tilde{g}_N(\tqW,\alpha) \asymp g_N(\alpha)$. Thus, we obtain an iterative deterministic equivalent
$g_N(\alpha)$ of $f_N((\tqH,\tqW),\alpha)$, i.e., $f_N((\tqH,\tqW),\alpha) \asymp g_N(\alpha)$.

When $\tqW$ is treated as a deterministic matrix, applying Lemma \ref{Lemma:TXXT} (see Appendix C), we have
\begin{equation}
 \frac{1}{N}\tr \qQ \left(\qB_N + \alpha \qI_N\right)^{-1} \asymp \frac{1}{N} \tr \qQ \left(\alpha \qI_N +   \alpha e \big(\qI_N-\beta\tqW^H\tqW \big)^2 \right)^{-1},  \label{eq:trQB1}
\end{equation}
where
\begin{align}
 e & = \frac{1}{N} \tr \qR_1 \left(\alpha \qI_K + \te \qR_1 \right)^{-1}, \label{eq:e2} \\
 \te & = \frac{1}{N} \tr \big(\qI_N-\beta\tqW^H\tqW \big)^2 \left( \qI_N + e \big(\qI_N-\beta\tqW^H\tqW \big)^2 \right)^{-1}.  \label{eq:te1}
\end{align}
Notice the fact that $(\tqW^H\tqW)^2 = \tqW^H\tqW$ so \eqref{eq:trQB1} and \eqref{eq:te1} can be written respectively as
\begin{align}
 \frac{1}{N}& \tr \qQ \left(\alpha \qI_{N} + \alpha e \big(\qI_N-\beta\tqW^H\tqW\big)^2 \right)^{-1} = \frac{1}{\alpha(\beta^2-2\beta)e} \frac{1}{N} \tr \qQ \big(\omega\qI_{N} + \tqW^H\tqW \big)^{-1}, \label{eq:trQB2}
\end{align}
and
\begin{align}
 \te = &\frac{1}{(\beta^2-2\beta)e} \frac{1}{N}\tr \big(\omega\qI_{N} + \tqW^H\tqW \big)^{-1} + \frac{1}{e} \frac{1}{N}\sum_{l=1}^L \tqw_l^H \big(\omega\qI_{N} + \tqW^H\tqW \big)^{-1} \tqw_l, \label{eq:te2}
\end{align}
where $\omega \triangleq \frac{1 + e}{(\beta^2-2\beta)e}$ and $\tqw_l$ denotes the $l-$th row of $\tqW$.

Next, we aim to derive the deterministic equivalents of the terms $\frac{1}{N} \tr \qQ (\omega\qI_{N} + \tqW^H\tqW )^{-1}$ and $\tqw_l^H (\omega\qI_{N} + \tqW^H\tqW )^{-1} \tqw_l$. Applying a result of the Haar matrix in Lemma \ref{Lemma:WW} (see Appendix C) to \eqref{eq:trQB2} and combing \eqref{eq:trQB1}, we immediately get \eqref{eq:trQB}. Then, we deal with  the deterministic equivalent of $\tqw_l^H \big(\omega\qI_{N} + \tqW^H\tqW \big)^{-1} \tqw_l$. According to the matrix inverse lemma (see, e.g., \cite[Lemma 2.1]{Bai-09}\footnote{\cite[Lemma 2.1]{Bai-09}: For any $\qA\in \bbC^{n\times n}$ and $\qq \in \bbC^n$ with $\qA$ and $\qA+\qq\qq^H$ invertible, we have $$\qq^H \left( \qA+\qq\qq^H \right)^{-1} = \frac{1}{1+\qq^H\qA^{-1}\qq} \qq^H\qA^{-1}.$$}), we find
\begin{align}
 \tqw_l^H \big( \omega \qI_{N} + \tqW^H\tqW \big)^{-1} \tqw_l = \frac{\tqw_l^H \big( \omega \qI_{N} + \tqW_{[l]}^H\tqW_{[l]} \big)^{-1} \tqw_l}{1 + \qw_l^H \big( \omega \qI_{N} + \tqW_{[l]}^H\tqW_{[l]} \big)^{-1} \tqw_l},
\end{align}
where $\tqW_{[l]} \triangleq [ \tqw_1, \ldots, \tqw_{l-1}, \tqw_{l+1}, \ldots, \tqw_L ]^H \in \bbC^{(L-1)\times N}$. Then, the trace lemma for isometric matrices
\cite{Debbah-03TIT,Couillet-12IT} gives us
\begin{align}
 \tqw_l^H \big(\omega\qI_{N} + \tqW_{[l]}^H\tqW_{[l]} \big)^{-1} \tqw_l  \asymp & \frac{1}{N-L} \tr \big(\qI_N-\tqW_{[l]}^H\tqW_{[l]}\big) \big(\omega\qI_{N} + \tqW_{[l]}^H\tqW_{[l]} \big)^{-1} \nonumber \\
 = & \frac{1+ \omega}{N-L} \tr \big(\omega\qI_{N} + \tqW_{[l]}^H \tqW_{[l]} \big)^{-1} - \frac{N}{N-L}.  \label{eq:wWlWlw}
\end{align}
Now, applying \cite[Lemma 2.2]{Bai-09} and \eqref{eq:lemmaQI} to \eqref{eq:wWlWlw}, we get
\begin{align}
 \tqw_l^H \big( \omega \qI_{N} + \tqW^H\tqW \big)^{-1} \tqw_l \asymp \frac{1}{\omega+1}.  \label{eq:wWWw}
\end{align}
Substituting \eqref{eq:wWWw} into \eqref{eq:te2} and using \eqref{eq:lemmaQI} and \eqref{eq:e2}, we obtain \eqref{eq:e}.
\end{proof}

Note that $m_{\qB_N, \qQ}$, $e$, $t_1$, and $t_2$ are all functions of $\alpha$ and $\beta$, but for ease of expression, $\alpha$ and $\beta$ are dropped.

Theorem \ref{Th: 1} indicates that $m_{\qB,\qQ}$ can be approximated by its deterministic equivalent $ \frac{ t_1 + t_2}{\alpha} \frac{1}{N} \tr \qQ$ without knowing the actual realization of channel random components. The deterministic equivalent is analytical and is much easier to compute than $\Ex_{\qB}\{ m_{\qB,\qQ} \}$, which requires time-consuming Monte-Carlo simulations. Motivated by this result in the large system limit, we aim to derive the deterministic equivalent of $\gamma_{k}$.

The SINR $\gamma_k$ in \eqref{eq:SINR} consists of three terms: $(\textrm{i})$ the signal power $| \qh_k^H ( \check{\qH}^H \check{\qH} + \alpha \qI_N)^{-1} \check{\qh}_k |^2$, $(\textrm{ii})$ the interference power $\qh_k^H ( \check{\qH}^H \check{\qH} + \alpha \qI_N )^{-1} \check{\qH}_{[k]}^H \qP_{[k]} \check{\qH}_{[k]} (\check{\qH}^H \check{\qH} + \alpha \qI_N )^{-1} \qh_k$, and $(\textrm{iii})$ the noise power $\nu$. Using Theorem \ref{Th: 1}, we establish the following three lemmas to derive the deterministic equivalent of each term, whose proofs are detailed in Appendices B-I, B-II, and B-III, successively.

\begin{lemma}\label{Lemma:Deterministic of signal power}
Under the assumption of Theorem \ref{Th: 1}, as $\calN \rightarrow \infty$, we have
\begin{equation}\label{eq:Deterministic of signal power}
 \qh_k^H \left( \check{\qH}^H \check{\qH} + \alpha \qI_N \right)^{-1} \check{\qh}_k \asymp \overline{a}_k,
\end{equation}
where $\overline{a}_k$ has been obtained by \eqref{eq:Dea}.
\end{lemma}

\begin{lemma}\label{Lemma:Deterministic of interference power}
Under the assumption of Theorem \ref{Th: 1}, as $\calN \rightarrow \infty$, we have
\begin{equation}\label{eq:Deterministic of interference power}
  \qh_k^H \left( \check{\qH}^H \check{\qH} + \alpha \qI_N \right)^{-1} \check{\qH}_{[k]}^H \check{\qH}_{[k]} \left( \check{\qH}^H \check{\qH} + \alpha \qI_N \right)^{-1} \qh_k \asymp \overline{b}_k,
\end{equation}
where $\overline{b}_k$ has been obtained by \eqref{eq:Deb}.
\end{lemma}

\begin{lemma}\label{Lemma:Deterministic of noise power}
Under the assumption of Theorem \ref{Th: 1}, as $\calN \rightarrow \infty$, we have
\begin{equation}\label{eq:Deterministic of noise power}
 \nu \asymp \overline{\nu} ,
\end{equation}
where $\overline{\nu}$ can be obtained by \eqref{eq:Dev} for Case I and by \eqref{eq:Dev2} for Case II.
\end{lemma}

According to Lemma \ref{Lemma:Deterministic of signal power}, Lemma \ref{Lemma:Deterministic of interference power}, and Lemma \ref{Lemma:Deterministic of noise
power}, we obtain the deterministic equivalent $\overline{\gamma}_k$ of $\gamma_k$ in \eqref{eq:gamma1 deterministic equivalent}. The proof is then completed.

\section*{Appendix B: Proofs of Lemma \ref{Lemma:Deterministic of signal power}, Lemma \ref{Lemma:Deterministic of interference power}, and Lemma \ref{Lemma:Deterministic of noise power}}

\subsection*{B-I: Proof of Lemma \ref{Lemma:Deterministic of signal power}}

We start from an application of the matrix inverse lemma \cite[Lemma 2.1]{Bai-09} to the signal term, which results in
\begin{align}
 \qh_k^H \left( \check{\qH}^H \check{\qH} + \alpha \qI_N \right)^{-1} \check{\qh}_k =  \frac{\qh_k^H \big( \check{\qH}_{[k]}^H \check{\qH}_{[k]} + \alpha \qI_N \big)^{-1} \check{\qh}_k}{1+\check{\qh}_k^H \big( \check{\qH}_{[k]}^H \check{\qH}_{[k]} + \alpha \qI_N \big)^{-1} \check{\qh}_k}. \label{eq:hcHcHch}
\end{align}
Using \cite[Lemma 2.3 and Lemma 2.2]{Bai-09}, we obtain
\begin{align}
  \qh_k^H \left( \check{\qH}_{[k]}^H \check{\qH}_{[k]} + \alpha \qI_N \right)^{-1} \check{\qh}_k
 \asymp & r_{1,k}\frac{1}{N} \tr \left( \check{\qH}^H \check{\qH} + \alpha \qI_N \right)^{-1} -  r_{1,k} \beta \frac{1}{N} \tr \qW^H \qW \left( \check{\qH}^H \check{\qH} + \alpha \qI_N \right)^{-1}. \label{eq:hBch}
\end{align}
Similarly,
\begin{align}
  & \check{\qh}_k^H \left( \check{\qH}_{[k]}^H \check{\qH}_{[k]} + \alpha \qI_N \right)^{-1} \check{\qh}_k  \nonumber \\
 \asymp & r_{1,k} \frac{1}{N} \tr \left( \check{\qH}^H \check{\qH} + \alpha \qI_N \right)^{-1} +  r_{1,k} (\beta^2-2\beta) \frac{1}{N} \tr \qW^H \qW \left( \check{\qH}^H \check{\qH} + \alpha \qI_N \right)^{-1}. \label{eq:chBch}
\end{align}
According to Theorem \ref{Th: 1}, we have
\begin{align}
 \frac{1}{N} \tr \left( \check{\qH}^H \check{\qH} + \alpha \qI_N \right)^{-1} \asymp \frac{t_1 + t_2 }{\alpha}. \label{eq:cHcH}
\end{align}
Noticing that $\qW^H\qW = \tqW^H\tqW$ and by using the same approach as \eqref{eq:te2}, we obtain
\begin{align}
 \frac{1}{N} \tr \qW^H \qW \big( \check{\qH}^H \check{\qH} + \alpha \qI_N \big)^{-1}
 \asymp & \frac{1}{\alpha(\beta^2-2\beta)e} \frac{1}{N} \sum_{l=1}^L \tqw_l^H \big(\omega\qI_{N} + \tqW^H\tqW \big)^{-1} \tqw_l  \asymp \frac{t_2}{\alpha}. \label{eq:WWcHcHI}
\end{align}
Substituting \eqref{eq:cHcH} and \eqref{eq:WWcHcHI} into \eqref{eq:hBch} and \eqref{eq:chBch}, we obtain
\begin{align}
 \qh_k^H \left( \check{\qH}_{[k]}^H \check{\qH}_{[k]} + \alpha \qI_N \right)^{-1} \check{\qh}_k &\asymp \frac{r_{1,k} \left( t_1 + t_2(1-\beta) \right)}{\alpha}, \label{eq:DehcHcHch}\\
 \check{\qh}_k^H \left( \check{\qH}_{[k]}^H \check{\qH}_{[k]} + \alpha \qI_N \right)^{-1} \check{\qh}_k &\asymp \frac{r_{1,k} \left( t_1 + t_2(1-\beta)^2 \right)}{\alpha}. \label{eq:DechcHcHch}
\end{align}
Consequently, the expression of \eqref{eq:hcHcHch}, together with \eqref{eq:DehcHcHch} and \eqref{eq:DechcHcHch}, yields \eqref{eq:Deterministic of signal power}.

\subsection*{B-II: Proof of Lemma \ref{Lemma:Deterministic of interference power}}

Using the fact that $\qA^{-1}-\qD^{-1}=-\qA^{-1}(\qA-\qD)\qD^{-1}$, we have
\begin{align}
  &\qh_k^H \left( \check{\qH}^H \check{\qH} + \alpha \qI_N \right)^{-1} \check{\qH}_{[k]}^H \check{\qH}_{[k]} \left( \check{\qH}^H \check{\qH} + \alpha \qI_N \right)^{-1} \qh_k  \nonumber\\
 =&\qh_k^H \left( \check{\qH}^H \check{\qH} + \alpha \qI_N \right)^{-1} \qh_k - \alpha \qh_k^H \left( \check{\qH}_{[k]}^H \check{\qH}_{[k]} + \alpha \qI_N \right)^{-1} \left( \check{\qH}^H \check{\qH} + \alpha \qI_N \right)^{-1} \qh_k\nonumber\\
 &- \qh_k^H \left( \check{\qH}^H \check{\qH} + \alpha \qI_N \right)^{-1} \check{\qh}_k \check{\qh}_k^H \left( \check{\qH}^H \check{\qH} + \alpha \qI_N \right)^{-1} \qh_k  \nonumber\\
 &+ \alpha \qh_k^H \left( \check{\qH}^H \check{\qH} + \alpha \qI_N \right)^{-1} \check{\qh}_k \check{\qh}_k^H \left( \check{\qH}_{[k]}^H \check{\qH}_{[k]} + \alpha \qI_N \right)^{-1} \left( \check{\qH}^H \check{\qH} + \alpha \qI_N \right)^{-1} \qh_k.  \label{eq:hchHIchHkchHIh}
\end{align}
Applying the matrix inverse lemma,  we obtain
\begin{align}
 & \qh_k^H \left( \check{\qH}^H \check{\qH} + \alpha \qI_N \right)^{-1} \qh_k \nonumber\\
 = & \qh_k^H \left( \check{\qH}_{[k]}^H \check{\qH}_{[k]} + \alpha \qI_N \right)^{-1} \qh_k - \frac{\qh_k^H \big( \check{\qH}_{[k]}^H \check{\qH}_{[k]} + \alpha \qI_N \big)^{-1} \check{\qh}_k \check{\qh}_k^H \big( \check{\qH}_{[k]}^H \check{\qH}_{[k]} + \alpha \qI_N \big)^{-1} \qh_k}{1+ \check{\qh}_k^H  \big( \check{\qH}_{[k]}^H \check{\qH}_{[k]} + \alpha \qI_N \big)^{-1} \check{\qh}_k}. \label{eq:hcHcHh}
\end{align}
Similarly,
\begin{align}
 &\qh_k^H \big( \check{\qH}_{[k]}^H \check{\qH}_{[k]} + \alpha \qI_N \big)^{-1} \left( \check{\qH}^H \check{\qH} + \alpha \qI_N \right)^{-1} \qh_k  \nonumber\\
 = & \qh_k^H \left( \check{\qH}_{[k]}^H \check{\qH}_{[k]} + \alpha \qI_N \right)^{-2} \qh_k - \frac{\qh_k^H \big( \check{\qH}_{[k]}^H \check{\qH}_{[k]} + \alpha \qI_N \big)^{-2} \check{\qh}_k \check{\qh}_k^H \big( \check{\qH}_{[k]}^H \check{\qH}_{[k]} + \alpha \qI_N \big)^{-1} \qh_k}{1+ \check{\qh}_k^H  \big( \check{\qH}_{[k]}^H \check{\qH}_{[k]} + \alpha \qI_N \big)^{-1} \check{\qh}_k},  \label{eq:hcHcHcHcHh}
\end{align}
and
\begin{align}
 &\check{\qh}_k^H \big( \check{\qH}_{[k]}^H \check{\qH}_{[k]} + \alpha \qI_N \big)^{-1} \big( \check{\qH}^H \check{\qH} + \alpha \qI_N \big)^{-1} \qh_k  \nonumber\\
 = & \check{\qh}_k^H \left( \check{\qH}_{[k]}^H \check{\qH}_{[k]} + \alpha \qI_N \right)^{-2} \qh_k - \frac{\check{\qh}_k^H \big( \check{\qH}_{[k]}^H \check{\qH}_{[k]} + \alpha \qI_N \big)^{-2} \check{\qh}_k \check{\qh}_k^H \big( \check{\qH}_{[k]}^H \check{\qH}_{[k]} + \alpha \qI_N \big)^{-1} \qh_k}{1+ \check{\qh}_k^H  \big( \check{\qH}_{[k]}^H \check{\qH}_{[k]} + \alpha \qI_N \big)^{-1} \check{\qh}_k}.\label{eq:chcHcHcHcHh}
\end{align}
According to Theorem \ref{Th: 1}, we have
\begin{align}
 \qh_k^H \left( \check{\qH}_{[k]}^H \check{\qH}_{[k]} + \alpha \qI_N \right)^{-1} \qh_k \asymp \frac{r_{1,k} \left( t_1 + t_2 \right)}{\alpha}. \label{eq:DehcHcHh}
\end{align}
Noticing that
\begin{align}
 \qh_k^H \left( \check{\qH}_{[k]}^H \check{\qH}_{[k]} + \alpha \qI_N \right)^{-2} \qh_k = - \frac{\partial}{\partial\alpha} \qh_k^H \left( \check{\qH}_{[k]}^H \check{\qH}_{[k]} + \alpha \qI_N \right)^{-1} \qh_k,\nonumber
\end{align}
we thus obtain
\begin{align}
 \qh_k^H \left( \check{\qH}_{[k]}^H \check{\qH}_{[k]} + \alpha \qI_N \right)^{-2} \qh_k \asymp - r_{1,k} \frac{\partial }{\partial\alpha}  \left(\frac{ t_1 + t_2 }{\alpha}\right). \label{eq:DehcHcH2h}
\end{align}
Similarly, combining \eqref{eq:DehcHcHch} and \eqref{eq:DechcHcHch} yields
\begin{align}
 \check{\qh}_k^H \left( \check{\qH}_{[k]}^H \check{\qH}_{[k]} + \alpha \qI_N \right)^{-2} \qh_k &\asymp - r_{1,k} \frac{\partial }{\partial\alpha}  \left(\frac{ t_1 + (1-\beta)t_2 }{\alpha}\right),  \label{eq:DechcHcH2h} \\
 \check{\qh}_k^H \left( \check{\qH}_{[k]}^H \check{\qH}_{[k]} + \alpha \qI_N \right)^{-2} \check{\qh}_k &\asymp - r_{1,k} \frac{\partial }{\partial\alpha}  \left(\frac{ t_1 + (1-\beta)^2 t_2 }{\alpha}\right).  \label{eq:DechcHcH2ch}
\end{align}
Substituting \eqref{eq:DehcHcHch}, \eqref{eq:DechcHcHch}, \eqref{eq:DehcHcHh}, \eqref{eq:DehcHcH2h}, \eqref{eq:DechcHcH2h}, and \eqref{eq:DechcHcH2ch} into \eqref{eq:hcHcHh}, \eqref{eq:hcHcHcHcHh}, and \eqref{eq:chcHcHcHcHh}, and combining \eqref{eq:Deterministic of signal power} and \eqref{eq:hchHIchHkchHIh}, we obtain \eqref{eq:Deterministic of interference power}.

\subsection*{B-III: Proof of Lemma \ref{Lemma:Deterministic of noise power}}

From \eqref{eq:nu}, we first have
\begin{align}
 & \frac{1}{N} \tr\left( \check{\qH}^H \check{\qH} + \alpha \qI_N \right)^{-1} \check{\qH}^H \check{\qH} \left( \check{\qH}^H \check{\qH} + \alpha \qI_N \right)^{-1} \nonumber \\
 = & \frac{1}{N} \tr \left( \check{\qH}^H \check{\qH} + \alpha \qI_N \right)^{-1} - \alpha \frac{1}{N} \tr \left( \check{\qH}^H \check{\qH} + \alpha \qI_N \right)^{-2},
\end{align}
which, together with Theorem \ref{Th: 1}, yields
\begin{align}
 \frac{1}{N} \tr \left( \check{\qH}^H \check{\qH} + \alpha \qI_N \right)^{-1} \check{\qH}^H \check{\qH} \left( \check{\qH}^H \check{\qH} + \alpha \qI_N \right)^{-1} \asymp \frac{\partial t_1}{\partial \alpha}+\frac{\partial t_2}{\partial \alpha} . \label{eq:cHcHcHcH}
\end{align}
For Case I, we have
\begin{align}
 &\qf_l^H \left( \check{\qH}^H \check{\qH} + \alpha \qI_N \right)^{-1} \check{\qH}^H \check{\qH} \left( \check{\qH}^H \check{\qH} + \alpha \qI_N \right)^{-1} \qf_l \nonumber \\
 = & \qf_l^H \left( \check{\qH}^H \check{\qH} + \alpha \qI_N \right)^{-1} \qf_l - \alpha \qf_l^H \left( \check{\qH}^H \check{\qH} + \alpha \qI_N \right)^{-2} \qf_l,  \label{eq:fcHcHcHcHf}
\end{align}
where $l = 1,\ldots,L$. From \eqref{eq:trQB1} and \eqref{eq:trQB2}, we obtain
\begin{align}
 \qf_l^H \left( \check{\qH}^H \check{\qH} + \alpha \qI_N \right)^{-1} \qf_l \asymp \frac{1}{\alpha(\beta^2-2\beta)e} \tr \qf_l \qf_l^H \left(\omega\qI_{N} + \qW^H\qW \right)^{-1}.
\end{align}
Noticing the fact that $\qW = (\qF\qF^H)^{-\frac{1}{2}}\qF$, by using the matrix inversion formula\footnote{
For invertible $\qA,\qB$ and $\qR$ matrices, suppose that $\qB=\qA+\qX\qR\qY$, then $\qB^{-1} = \qA^{-1}-\qA^{-1}\qX(\qR^{-1}+\qY\qA^{-1}\qX)^{-1}\qY\qA^{-1}$.}, we obtain
\begin{align}
 \tr \qf_l \qf_l^H \left(\omega\qI_{N} + \qW^H\qW \right)^{-1}
 = & \tr  \qf_l \qf_l^H \left(\omega^{-1}\qI_{N} - \omega^{-1} \qF^H \left(\qF\qF^H + \omega^{-1} \qF \qF^H\right)^{-1} \qF \omega^{-1} \right)  \nonumber \\
 = & \frac{1}{\omega+1} \tr \qf_l \qf_l^H   \asymp  \frac{r_{2,l}}{\omega+1}.
\end{align}
As a result,
\begin{align}
 \qf_l^H \left( \check{\qH}^H \check{\qH} + \alpha \qI_N \right)^{-1} \qf_l \asymp \frac{r_{2,l} t_2}{c_2\alpha}.
\end{align}
Combining this with \eqref{eq:fcHcHcHcHf}, we obtain
\begin{align}
 &\qf_l^H \left( \check{\qH}^H \check{\qH} + \alpha \qI_N \right)^{-1} \check{\qH}^H \check{\qH} \left( \check{\qH}^H \check{\qH} + \alpha \qI_N \right)^{-1} \qf_l \asymp \frac{r_{2,l}}{c_2} \frac{\partial t_2}{\partial \alpha}. \label{eq:fcHcHcHcHf2}
\end{align}
From \eqref{eq:cHcHcHcH} and \eqref{eq:fcHcHcHcHf2}, the proof of \eqref{eq:Dev} can be accomplished using
\begin{subequations}\label{eq:m_Exm}
\begin{align}
 &\frac{1}{N} \tr \big( \check{\qH}^H \check{\qH} + \alpha \qI_N \big)^{-1} \check{\qH}^H \check{\qH} \big( \check{\qH}^H \check{\qH} + \alpha \qI_N \big)^{-1}\nonumber \\
 \asymp & \Ex \left\{ \frac{1}{N} \tr \big( \check{\qH}^H \check{\qH} + \alpha \qI_N \big)^{-1} \check{\qH}^H \check{\qH} \big( \check{\qH}^H \check{\qH} + \alpha \qI_N \big)^{-1} \right\},  \\
 &\qf_l^H \big( \check{\qH}^H \check{\qH} + \alpha \qI_N \big)^{-1} \check{\qH}^H \check{\qH} \big( \check{\qH}^H \check{\qH} + \alpha \qI_N \big)^{-1} \qf_l \nonumber \\
 \asymp & \Ex \left\{ \qf_l^H \big( \check{\qH}^H \check{\qH} + \alpha \qI_N \big)^{-1} \check{\qH}^H \check{\qH} \big( \check{\qH}^H \check{\qH} + \alpha \qI_N \big)^{-1} \qf_l \right\}.
\end{align}
\end{subequations}
By using the martingale approach, we can prove \eqref{eq:m_Exm} (See \cite{WenCK-13TIT} for a similar application).

Similarly, for Case II, we have
\begin{align}
 \Ex \left\{ \tr \qF \big( \check{\qH}^H \check{\qH} + \alpha \qI_N \big)^{-1} \check{\qH}^H \check{\qH} \big( \check{\qH}^H \check{\qH} + \alpha \qI_N \big)^{-1} \qF^H \right\} \asymp \frac{\tr \qR_2}{c_2} \frac{\partial t_2}{\partial \alpha}. \label{eq:FcHcHcHcHF}
\end{align}
Therefore, we obtain \eqref{eq:Dev2}.

\section*{Appendix C: Related Lemmas}\label{Appendix:Related Lemmas}

In this appendix, we provide some lemmas needed in the proof of Appendix A.

\begin{lemma}\label{Lemma:TXXT}
Let $\qX \equiv [\frac{1}{\sqrt{N}}X_{ij}] \in \bbC^{N \times K}$, where $X_{ij}$'s are i.i.d. with zero mean, unit variance and finite $4$-th order moment. In addition, let $\qQ \in \bbC^{N
\times N}$, $\qT \in \bbC^{N\times N}$, and $\qR \in \bbC^{K\times K}$ be nonnegative definite matrices with uniformly bounded spectral norm (with respect to $N$, $N$, and $K$, respectively).
Consider an $N \times N$ matrix of the form $\qB_N = \qT^\frac{1}{2} \qX \qR \qX^H \qT^\frac{1}{2}$. Define $c_1 = N/K$. Then, as $K, N \rightarrow \infty$ such that $0<\lim\inf_N c_1
\leq \lim\sup_N c_1 <\infty$, the following holds for any $\omega \in \bbR^+$:
\begin{equation}
 \frac{1}{N} \tr \qQ \big( \qB_N + \omega \qI_{N} \big)^{-1} \asymp \frac{1}{N} \tr \qQ \left(\omega \qI_{N} + \omega e \qT\right)^{-1},
\end{equation}
where $e$ is given as the unique solution to the fixed-point equations
\begin{align}
 e &= \frac{1}{N} \tr \qR (\omega \qI_K +  \te \qR )^{-1}, \nonumber \\
 \te &= \frac{1}{N} \tr \qT ( \qI_N +  e \qT )^{-1}.  \nonumber
\end{align}
\end{lemma}
\begin{proof}
As a special case of \cite[Theorem 1]{ZhangJun-13JSAC} or \cite[Theorem 1]{Wagner-12IT}, the result can be obtained immediately.
\end{proof}

\begin{lemma}\label{Lemma:WW}
Let $\qQ \in \bbC^{N \times N}$ be a nonnegative definite matrix with uniformly bounded spectral norm (with respect to $N$) and $\tqW \in \bbC^{L\times N}$ be
$L \leq N$ rows of an $N \times N$ Haar-distributed unitary random matrix. Define $c_2 = L/N$. Then, as $L, N \rightarrow \infty$ such that $0<\lim\inf_N c_2 \leq
\lim\sup_N c_2 \leq 1$, the following holds for any $\omega \in \bbR^+$:
\begin{equation}
 \frac{1}{N} \tr \qQ \big( \tqW^H \tqW + \omega \qI_{N} \big)^{-1} \asymp \left( \frac{c_2}{\omega+1} + \frac{1-c_2}{\omega} \right) \frac{1}{N} \tr \qQ. \label{eq:trQWW}
\end{equation}
\end{lemma}
\begin{proof}
Since $\tqW^H \tqW = \qI_L$ for $c_2 = 1$ (i.e., $N = L$), \eqref{eq:trQWW} evidently holds. We assume $c_2 < 1$ in the following proof.
Firstly, we consider a special case with $\qQ = \qI$. Using the identity of the Stieltjes transform \cite[Lemma 3.1]{Couillet-11BOOK} \footnote{\cite[Lemma 3.1]{Couillet-11BOOK}: Let $\qA \in \bbC^{N\times n}$, $\qB \in \bbC^{n\times N}$, such that $\qA\qB$ is Hermitian. Then, for $z \in \bbC \backslash \bbR$ $$\frac{n}{N} m_{\qB\qA}(z) = m_{\qA\qB}(z) + \frac{N-n}{N} \frac{1}{z}.$$}, we have
\begin{align}
 \frac{1}{N} \tr \big( \tqW^H \tqW + \omega \qI_{N} \big)^{-1} = \frac{c_2}{L} \tr \big( \tqW \tqW^H + \omega \qI_L \big)^{-1} + \frac{1-c_2}{\omega}.
\end{align}
Notice that the rows of $\tqW$ are orthogonal and hence $\tqW \tqW^H = \qI_L$. Therefore,
\begin{align}
 \frac{1}{N} \tr \big( \tqW^H \tqW + \omega \qI_{N} \big)^{-1} \asymp  \delta \triangleq \frac{c_2}{\omega+1} + \frac{1-c_2}{\omega}.  \label{eq:lemmaQI}
\end{align}

Next, for any nonnegative definite matrix with uniformly bounded spectral norm (with respect to $N$) $\qQ$, we have
\begin{align}
   &\frac{1}{N} \tr \qQ \big( \tqW^H \tqW + \omega \qI_N \big)^{-1} - \delta \frac{1}{N} \tr \qQ    \nonumber \\
 = & \left( 1 - \delta \omega \right) \frac{1}{N} \tr \qQ \big( \tqW^H \tqW + \omega \qI_N \big)^{-1} - \delta \sum_{l=1}^L \qw_l^H \qQ \big( \tqW^H \tqW + \omega \qI_N \big)^{-1} \qw_l, \label{eq:lemmaQ1}
\end{align}
where the first equality follows from the resolvent identity: $\qA^{-1}-\qB^{-1} = \qA^{-1} (\qB - \qA) \qB^{-1}$ for invertible matrices $\qA$ and $\qB$. Using the matrix inverse lemma \cite[Lemma 2.1]{Bai-09}, the trace lemma for isometric matrices \cite{Debbah-03TIT,Couillet-12IT}, and the fact that $\qQ$ has uniformly bounded spectral norm (with respect to $N$), we obtain
\begin{align}
  \sum_{l=1}^L \qw_l^H \qQ \big( \tqW^H \tqW + \omega \qI_N \big)^{-1} \qw_l
   = & \sum_{l=1}^L \frac{\qw_l^H \qQ \big( \tqW^H_{[l]} \tqW_{[l]} + \omega \qI_N \big)^{-1} \qw_l}{1+\qw_l^H \big( \tqW^H_{[l]} \tqW_{[l]} + \omega \qI_N \big)^{-1} \qw_l} \nonumber \\
 \asymp & c_2 \frac{ (1+\omega)\frac{1}{N} \tr \qQ \big( \tqW^H \tqW + \omega \qI_N \big)^{-1} - \frac{1}{N} \tr \qQ}{ (1+\omega)\frac{1}{N} \tr \big( \tqW^H \tqW + \omega \qI_N \big)^{-1} - c_2}. \label{eq:sumwQw}
\end{align}
Substituting \eqref{eq:sumwQw} into \eqref{eq:lemmaQ1}, and combining \eqref{eq:lemmaQI}, yields
\begin{align}
 \frac{(1+\omega)\delta \left( \frac{1}{N} \tr \qQ \big( \tqW^H \tqW + \omega \qI_N \big)^{-1} - \delta \frac{1}{N} \tr \qQ \right)}{(1+\omega)\delta -c_2} \asymp 0.
\end{align}
Therefore, we get \eqref{eq:trQWW}.
\end{proof}



\bibliographystyle{IEEEtran}

\begin{thebibliography}{10}
\providecommand{\url}[1]{#1} \csname url@samestyle\endcsname \providecommand{\newblock}{\relax} \providecommand{\bibinfo}[2]{#2}
\providecommand{\BIBentrySTDinterwordspacing}{\spaceskip=0pt\relax} \providecommand{\BIBentryALTinterwordstretchfactor}{4}
\providecommand{\BIBentryALTinterwordspacing}{\spaceskip=\fontdimen2\font plus \BIBentryALTinterwordstretchfactor\fontdimen3\font minus
  \fontdimen4\font\relax}
\providecommand{\BIBforeignlanguage}[2]{{%
\expandafter\ifx\csname l@#1\endcsname\relax
\typeout{** WARNING: IEEEtran.bst: No hyphenation pattern has been}%
\typeout{** loaded for the language `#1'. Using the pattern for}%
\typeout{** the default language instead.}%
\else \language=\csname l@#1\endcsname \fi #2}}
\providecommand{\BIBdecl}{\relax} \BIBdecl

\bibitem{Mitola-99PCM}
{J. Mitola and G. Q. Maguire}, ``{Cognitive radio: Making software radio more
  personal},'' \emph{IEEE Personal Commun. Mag.}, vol.~6, no.~4, pp. 13--18,
  Aug. 1999.

\bibitem{Letaief-09ProcIEEE}
{K. B. Letaief and W. Zhang}, ``{Cooperative communications for cognitive radio
  networks},'' \emph{Proceedings of the IEEE}, vol.~97, no.~5, pp. 878--893,
  May 2009.

\bibitem{Goldsmith-09IProc}
{A. Goldsmith, S. Jafar, I. Maric, and S. Srinivasa}, ``{Breaking spectrum
  gridlock with cognitive radios: An information theoretic perspective},''
  \emph{Proceedings of the IEEE}, vol.~97, no.~5, pp. 894--914, May 2009.

\bibitem{Liang-11TVT}
{Y.-C. Liang, K.-C. Chen, G. Y. Li, and P. Mahonen}, ``{Cognitive radio
  networking and communications: An overview},'' \emph{IEEE Trans. Veh.
  Technol.}, vol.~60, no.~7, pp. 3386--3407, Sep. 2011.

\bibitem{Kim-11TVT}
{J. Kim, Y. Shin, T. W. Ban, and R. Schober}, ``{Effect of spectrum sensing
  reliability on the capacity of multiuser uplink cognitive radio systems},''
  \emph{IEEE Trans. Veh. Technol.}, vol.~60, no.~9, pp. 4349--4362, Nov. 2011.

\bibitem{Zheng-09TSP}
{G. Zheng, K.-K. Wong, and B. Ottersten}, ``{Robust cognitive beamforming with
  bounded channel uncertainties},'' \emph{IEEE Trans. Sig. Proc.}, vol.~57,
  no.~12, pp. 4871--4881, Dec. 2009.

\bibitem{Huang-11JSAC}
{S. Huang, X. Liu, and Z. Ding}, ``{Decentralized cognitive radio control based
  on inference from primary link control information},'' \emph{IEEE J. Sel.
  Areas Commun.}, vol.~29, no.~2, pp. 394--406, Feb. 2011.

\bibitem{ZhangR-08JSTSP}
{R. Zhang and Y.-C. Liang}, ``{Exploiting multi-antennas for opportunistic
  spectrum sharing in cognitive radio networks},'' \emph{IEEE J. Sel. Topics
  Sig. Proc.}, vol.~2, no.~1, pp. 88--102, Feb. 2008.

\bibitem{Hamdi-09TWC}
{K. Hamdi, W. Zhang, and K. B. Letaief}, ``{Opportunistic spectrum sharing in
  cognitive MIMO wireless networks},'' \emph{IEEE Trans. Wireless Commun.},
  vol.~8, no.~8, pp. 4098--4109, Aug. 2009.

\bibitem{HeYY-13TWC}
{Y. Y. He and S. Dey}, ``{Sum rate maximization for cognitive MISO broadcast
  channels: Beamforming design and large systems analysis},'' \emph{IEEE Trans.
  Wireless Commun.}, vol.~13, no.~5, pp. 2383--2401, May 2014.

\bibitem{HeYY-13ICASSP}
------, ``{Weighted sum rate maximization for cognitive MISO broadcast channel:
  Large system analysis},'' in \emph{Proc. ICASSP}, Vancouver, Canada, May
  2013, pp. 4873--4877.

\bibitem{Chen-13TVT}
{X. Chen and C. Yuen}, ``{Efficient resource allocation in rateless coded
  MU-MIMO cognitive radio network with QoS provisioning and limited
  feedback},'' \emph{IEEE Trans. Veh. Technol.}, vol.~62, no.~1, pp. 395--399,
  Jan. 2013.

\bibitem{Caire-03TIT}
{G. Caire and S. Shamai}, ``{On the achievable throughput of a multiantenna
  Gaussian broadcast channel},'' \emph{IEEE Trans. Inf. Theory}, vol.~49,
  no.~7, pp. 1691--1706, Jul. 2003.

\bibitem{Samardzija-07ICC}
{D. Samardzija, H. Huang, T. Sizer, and R. Valenzuela}, ``{Experimental
  downlink multiuser MIMO system with distributed and coherently coordinated
  transmit antennas},'' in \emph{Proc. IEEE Int. Conf. on Commun. (ICC)},
  Glasgow, Scotland, Jun. 2007, pp. 5365--5370.

\bibitem{Irmer-09COMMag}
{R. Irmer, H.-P. Mayer, A. Weber, V. Braun, M. Schmidt, M. Ohm, N. Ahr, A.
  Zoch, C. Jandura, P. Marsch, and G. Fettweis}, ``{Multisite field trial for
  LTE and advanced concepts},'' \emph{IEEE Commun. Mag.}, vol.~47, no.~2, pp.
  92--98, Feb. 2009.

\bibitem{Suraweera-13ICC}
{H. A. Suraweera, H. Q. Ngo, T. Q. Duong, C. Yuen, and E. G. Larsson},
  ``{Multi-pair amplify-and-forward relaying with large antenna arrays},'' in
  \emph{Proc. IEEE Int. Conf. on Commun. (ICC)}, Budapest, Hungary, Jun. 2013,
  pp. 3228--3233.

\bibitem{Joham-02ISSSTA}
{M. Joham, K. Kusume, M. H. Gzara, W. Utschick, and J. A. Nossek}, ``{Transmit
  Wiener filter for the downlink of TDDDS-CDMA systems},'' in \emph{Proc. IEEE
  7th Int. Symp. Spread-Spectrum Tech. Appl. (ISSSTA)}, vol.~1, 2002, pp.
  9--13.

\bibitem{Peel-05Tcom}
{C. B. Peel, B. M. Hochwald, and A. L. Swindlehurst}, ``{A vector-perturbation
  technique for near-capacity multiantenna multiuser communication--Part I:
  Channel inversion and regularization},'' \emph{IEEE Trans. Commun.}, vol.~53,
  no.~1, pp. 195--202, Jan. 2005.

\bibitem{Nguyen-08GLCOM}
{V. K. Nguyen and J. S. Evans}, ``{Multiuser transmit beamforming via
  regularized channel inversion: A large system analysis},'' in \emph{Proc.
  IEEE Global Commun. Conf. (GLOBECOM)}, New Orleans, LA, Dec. 2008, pp. 1--4.

\bibitem{Muharar-11ICC}
{R. Muharar and J. Evans}, ``{Downlink beamforming with transmit-side channel
  correlation: A large system analysis},'' in \emph{Proc. IEEE Int. Conf. on
  Commun. (ICC)}, Kyoto, Japan, Jun. 2011, pp. 1--5.

\bibitem{Wagner-12IT}
{S. Wagner, R. Couillet, M. Debbah, and D. T. M. Slock}, ``{Large system
  analysis of linear precoding in correlated MISO broadcast channels under
  limited feedback},'' \emph{IEEE Trans. Inf. Theory}, vol.~58, no.~7, pp.
  4509--4537, Jul. 2012.

\bibitem{Muharar-13TCom}
{R. Muharar, R. Zakhour, and J. Evans}, ``{Optimal power allocation and user
  loading for multiuser MISO channels with regularized channel inversion},''
  \emph{IEEE Trans. Commun.}, vol.~61, no.~12, pp. 5030--5041, Dec. 2013.

\bibitem{Geraci-13JSAC}
{G. Geraci, R. Couillet, J. Yuan, M. Debbah, and I. B. Collings}, ``{Large
  system analysis of linear precoding in MISO broadcast channels with
  confidential messages},'' \emph{IEEE J. Sel. Areas Commun.}, vol.~31, no.~9,
  pp. 1660--1671, Sep. 2013.

\bibitem{ZhangJun-14CL}
{J. Zhang, C. Yuen, C.-K. Wen, S. Jin, and X. Q. Gao}, ``{Ergodic secrecy
  sum-rate for multiuser downlink transmission via regularized channel
  inversion: Large system analysis},'' \emph{IEEE Commun. Letters}, vol.~18,
  no.~9, pp. 1627--1630, Sep. 2014.

\bibitem{Muharar-12ISIT}
{R. Muharar, R. Zakhour, and J. Evans}, ``{Base station cooperation with
  limited feedback: A large system analysis},'' in \emph{Proc. IEEE Int. Symp.
  on Inf. Theory Proceedings (ISIT)}, Cambridge, MA, Jul. 2012, pp. 1152--1156.

\bibitem{Huang-13TWC}
{Y. Huang, C. W. Tan, and B. D. Rao}, ``{Joint beamforming and power control in
  coordinated multicell: Max-min duality, effective network and large system
  transition},'' \emph{IEEE Trans. Wireless Commun.}, vol.~12, no.~6, pp.
  2730--2742, Jun. 2013.

\bibitem{ZhangJun-13TWC}
{J. Zhang, C.-K. Wen, S. Jin, X. Q. Gao, and K.-K. Wong}, ``{Large system
  analysis of cooperative multi-cell downlink transmission via regularized
  channel inversion with imperfect CSIT},'' \emph{IEEE Trans. Wireless
  Commun.}, vol.~12, no.~10, pp. 4801--4813, Oct. 2013.

\bibitem{WenCK-14TWC}
{C.-K. Wen, J.-C. Chen, K.-K. Wong, and P. Ting}, ``{Message passing algorithm
  for distributed downlink regularized zero-forcing beamforming with
  cooperative base stations},'' \emph{IEEE Trans. Wireless Commun.}, vol.~13,
  no.~5, pp. 2920--2930, May 2014.

\bibitem{Couillet-11BOOK}
{R. Couillet and M. Debbah}, \emph{{Random Matrix Methods for Wireless
  Communications}}.\hskip 1em plus 0.5em minus 0.4em\relax Cambridge University
  Press, 2011.

\bibitem{ZhangR-09TWC}
{R. Zhang}, ``{On peak versus average interference power constraints for
  protecting primary users in cognitive radio networks},'' \emph{IEEE Trans.
  Wireless Commun.}, vol.~8, no.~4, pp. 2112--2120, Apr. 2009.

\bibitem{Wang-09TWC}
{C.-X. Wang, X. Hong, H.-H. Chen, and J. Thompson}, ``{On capacity of cognitive
  radio networks with average interference power constraints},'' \emph{IEEE
  Trans. Wireless Commun.}, vol.~8, no.~4, pp. 1620--1625, Apr. 2009.

\bibitem{LZhang-09TWC}
{L. Zhang, Y. Xin, and Y.-C. Liang}, ``{Weighted sum rate optimization for
  cognitive radio MIMO broadcast channels},'' \emph{IEEE Trans. Wireless
  Commun.}, vol.~8, no.~6, pp. 2950--2957, Jun. 2009.

\bibitem{Dai-13JSAC}
{L. Dai, Z. Wang, and Z. Yang}, ``{Spectrally efficient time-frequency training
  OFDM for mobile large-scale MIMO systems},'' \emph{IEEE J. Sel. Areas
  Commun.}, vol.~31, no.~2, pp. 251--263, Feb. 2013.

\bibitem{Gao-14CL}
{Z. Gao, L. Dai, Z. Lu, C. Yuen, Z. Wang}, ``{Super-resolution sparse MIMO-OFDM
  channel estimation based on spatial and temporal correlations},'' \emph{IEEE
  Commun. Letters}, vol.~18, no.~7, pp. 1266--1269, Jul. 2014.

\bibitem{Horn-90BOOK}
{R. Horn and C. Johnson}, \emph{{Matrix Analysis}}.\hskip 1em plus 0.5em minus
  0.4em\relax Cambridge Univ. Press, 1990.

\bibitem{Hachem-07AAP}
{W. Hachem, P. Loubaton, and J. Najim}, ``{Deterministic equivalents for
  certain functionals of large random matrices},'' \emph{Annals of Applied
  Probability}, 2007.

\bibitem{Billingsley-11BOOK}
{P. Billingsley}, \emph{{Probability and Measure}}.\hskip 1em plus 0.5em minus
  0.4em\relax Hoboken, NJ: Wiley, 2011.

\bibitem{Mar-67}
{V. A. Mar\v{c}enko and L. A. Pastur}, ``{Distributions of eigenvalues for some
  sets of random matrices},'' \emph{Math. USSR-Sbornik}, vol.~1, pp. 457--483,
  Apr. 1967.

\bibitem{Rudelson-09Arxiv}
{M. Rudelson and R. Vershynin}, ``{The smallest singular value of a random
  rectangular matrix},'' \emph{Commun. Pure Appl. Math.}, vol.~62, no.~12, pp.
  1595--1739, 2009.

\bibitem{SilversteinBai-95}
{J. W. Silverstein and Z. Bai}, ``{On the empirical distribution of eigenvalues
  of a class of large dimensional random matrices},'' \emph{J. Multiv. Anal.},
  vol.~54, pp. 175--192, 1995.

\bibitem{Hoydis-11}
\BIBentryALTinterwordspacing
{J. Hoydis, R. Couillet, and M. Debbah}, ``{Iterative deterministic equivalents
  for the performance analysis of communication systems},'' 2011. [Online].
  Available: \url{http://arxiv.org/abs/1112.4167.}
\BIBentrySTDinterwordspacing

\bibitem{Bai-09}
{Z. Bai, Y. Chen, and Y.-C. Liang}, \emph{{Random Matrix Theory and its
  Applications}}.\hskip 1em plus 0.5em minus 0.4em\relax World Scientific
  Publishing Company, 2009.

\bibitem{Debbah-03TIT}
{M. Debbah, W. Hachem, P. Loubaton, and M. de Courville}, ``{MMSE analysis of
  certain large isometric random precoded systems},'' \emph{IEEE Trans. Inf.
  Theory}, vol.~49, no.~5, pp. 1293--1311, May 2003.

\bibitem{Couillet-12IT}
{R. Couillet, J. Hoydis, and M. Debbah}, ``{Random beamforming over
  quasi-static and fading channels: A deterministic equivalent approach},''
  \emph{IEEE Trans. Inf. Theory}, vol.~58, no.~10, pp. 6392--6425, Oct. 2012.

\bibitem{WenCK-13TIT}
{C.-K. Wen, G. Pan, K.-K. Wong, M. H. Guo, and J. C. Chen}, ``{A deterministic
  equivalent for the analysis of non-Gaussian correlated MIMO multiple access
  channels},'' \emph{IEEE Trans. Inf. Theory}, vol.~59, no.~1, pp. 329--351,
  Jan. 2013.

\bibitem{ZhangJun-13JSAC}
{J. Zhang, C.-K. Wen, S. Jin, X. Q. Gao, and K.-K. Wong}, ``{On capacity of
  large-scale MIMO multiple access channels with distributed sets of correlated
  antennas},'' \emph{IEEE J. Sel. Areas Commun.}, vol.~31, no.~2, pp. 133--148,
  Feb. 2013.
\end{thebibliography}

\end{document}